\documentclass[11pt]{article}

\usepackage{geometry}
\usepackage{titlesec}
\usepackage{abstract}
\usepackage{xfrac}
\usepackage[T1]{fontenc}
\usepackage{helvet}
\usepackage[utf8]{inputenc}
\usepackage{authblk}
\usepackage{graphics}
\usepackage{circuitikz}

\providecommand{\keywords}[1]{{{\bf\textsc{Keywords: }}} #1}
\titleformat{\section}{\centering\large\bf\scshape}{\thesection}{1em}{}
\titleformat{\subsection}{\centering\large\bf\scshape}{\thesubsection}{1em}{}
\titleformat{\subsubsection}{\normalsize\bf\scshape}{\thesubsubsection}{1em}{}

\usepackage{hyperref}
\hypersetup{
	colorlinks=true,
	linkcolor=blue,
	filecolor=blue,      
	urlcolor=blue,
	citecolor=blue
}

\usepackage{natbib}
\usepackage{graphicx} 
\graphicspath{{art/}}

\usepackage{setspace}
\usepackage{xr}
\usepackage{color}
\usepackage{amsmath, amssymb}
\usepackage{amsthm}
\usepackage{amsfonts,dsfont,mathrsfs}
\usepackage{booktabs,multirow,array}
\usepackage{bm}

\usepackage{algorithm}
\usepackage[noend]{algorithmic}



\DeclareMathOperator{\mean}{\mathbb{E}}

\DeclareMathOperator{\var}{Var}
\DeclareMathOperator{\rvar}{\overline{Var}}

\newcommand\norm[1]{{\left\Vert{#1}\right\Vert}}

\newcommand\smallfrac[2]{\mbox{\small$\displaystyle\frac{#1}{#2}$}}

\newcommand\tsum{\textstyle\sum}

\newcommand\diff{\mathrm{d}}

\newcommand\ind{\mathbb{I}}

\newcommand\prob{\mathbb{P}}

\newcommand\Beta{\mathrm{B}}

\newcommand\BR{\mathbb{R}}

\newcommand\al{\alpha}
\newcommand\be{\beta}
\newcommand\ga{\gamma}    \newcommand\Ga{\Gamma}

  \newcommand\vep{\varepsilon}

\newcommand\la{\lambda}   

\newcommand\te{\theta}

\newcommand\atwo{\theta}
\newcommand\Spd{S_{PD}}

\allowdisplaybreaks[4]

\usepackage[normalem]{ulem} 
\newtheoremstyle{underlined}
  {3pt}      
  {3pt}      
  {\itshape} 
  {}         
  {\bf}
  {.}        
  {.5em}     
  {\uline{\scshape{\thmname{#1}\ \thmnumber{#2}\thmnote{ (#3)}}}}

\theoremstyle{underlined}

\newcounter{mycount}
\newtheorem{proposition}[mycount]{Proposition}
\newtheorem{corollary}[mycount]{Corollary}
\newtheorem{lemma}[mycount]{Lemma}

\begin{document}

\title{\bf\scshape\LARGE{Nonparametric Bayesian inference for the Gini--Simpson index}}

\author[1]{Pier Giovanni Bissiri\thanks{pier.bissiri@unimib.it}}
\author[1]{Riccardo Corradin\thanks{riccardo.corradin@unimib.it}}
\author[1]{Andrea Ongaro\thanks{andrea.ongaro@unimib.it}}
\affil[1]{\normalsize{Department of Economics Management and Statistics, University of Milano-Bicocca, Piazza dell’Ateneo Nuovo 1, 20126, Milan,
Italy}}

\date{}
\maketitle 

\begin{abstract}
	Many statistical problems concern the analysis of species distributions or, more generally, of discrete labeled quantities. Assessing species diversity constitutes a key step toward understanding population structure, and the Gini--Simpson index is among the most widely adopted diversity measures. In this manuscript, we examine several well-established nonparametric prior models for species frequencies and compare them with a newly proposed distribution within this framework. Specifically, we demonstrate that the conventional symmetric Dirichlet distribution leads to certain undesirable properties in terms of expectation and dispersion. These limitations can be mitigated by adopting an alternative symmetric Dirichlet specification, in which the parameter depends on the number of species. This modified formulation is characterized by analytical tractability and interpretability of its main summaries. Furthermore, when the number of distinct species is infinite, the classical Ferguson Dirichlet process exhibits unsatisfactory behavior compared to the more general Poisson--Dirichlet model. Notably, within this latter model, the posterior mean of the diversity index can be expressed as a convex combination of the optimal classical unbiased estimator and the prior expectation. Theoretical results are further supported by asymptotic analyses and systematically compared with their classical counterparts.
	
	\vspace{12pt}
	\noindent\keywords{Gini--Simpson index, Species diversity, Bayesian nonparametric, Posterior asymptotic}
\end{abstract}

\section{Introduction}

Species diversity is the variety and abundance of species in a population, strictly related to its structure. Evaluating the diversity is fundamental in several field, including engineering, physics~\citep{Let06}, ecology~\citep[e.g.][]{Mag04, Got13}, genetics~\citep{Far95} and neuroscience~\citep{Shl07}. Species diversity can be partitioned in two main source of information that we want to infer from a specific sample: the number of species in the population (\textit{species richness}) and the variability in the species abundances (\textit{evenness}).  In particular, evenness reaches its maximum if all species frequencies are equal, and it is minimum if we observe only a single species. Hence, suitable indexes for diversity should vary as far as both the number of species and evenness change. Among possible indexes, two diversity measures are commonly used in literature, namely the Shannon index \citep[e.g.][]{Sha48, Pie67, Mag04, Got13} and the Gini--Simpson index, here denoted by $S$. This manuscript exploits the latter, which has been introduced independently by \citet{Gin12} and \citet{Sim49}. The Gini--Simpson index is nothing but the probability to sample with replacement, from a population, two units belonging to two distinct species of categories. In formula, the Gini--Simpson index equals
\begin{equation}\label{eq:gsind}
S = 1 - \sum_{j=1}^k P_j^2, 
\end{equation}
where $k$ is the number of distinct species, $(P_1, \dots, P_k) \in \bigtriangleup^{k-1}$, and $P_j$ denotes the relative frequency of the $j$th species. Here,
\[
	\bigtriangleup^{k-1}	= \left\{ x_1, \dots, x_{k} \Big\vert x_j \geq 0,\; j =1, \dots, k;\; \sum_{j=1}^{k} x_j = 1\right\}
\]
denotes the $(k-1)$-dimensional simplex space. Frequentist estimators have been proposed in literature to perform inference about the Gini--Simpson index \citep[e.g.][]{Smi77}, but those estimators are showing a large variability, and Bayesian approaches are usually preferred since they smooth and regularize empirical estimates \citep{Gil79, Goo53}. In this framework, specify a prior distribution for the Gini--Simpson index $S$ corresponds to set a distributional family for the sequence of relative frequencies $\{P_j\}_{j=1}^k$ and possibly for $k$, if assumed finite but random. 

In the following sections, within a general Bayesian framework, we conduct an analysis and comparison of the inferential implications for the Gini--Simpson index $S$ arising from the adoption of several fundamental nonparametric prior models for species frequencies, along with a model specification that has not yet been explored in the context of species diversity. Accordingly, two primary settings are considered. First, we address the case in which the total number of species $k$ is finite but possibly unknown, a scenario commonly assumed in most applied studies. Secondly, we examine species diversity under the assumption of some common nonparametric priors that accommodate an infinite number of species. 

In the first setting, within a Bayesian approach, it is natural to assign a prior distribution to the total number of species $k$, which we let arbitrary at this stage. The model is then fully specified by making a specific conditional distribution assumption for the frequencies $(P_1, \dots, P_k)$, given the number of species $k$. The most diffuse choice in literature is the symmetric Dirichlet (\textsc{SD}) distribution, where the only parameter $\theta$ indexing the model is fixed, i.e. does not depend on $k$. This leads to the well-known class of Gibbs--type prior with finitely many species \citep[see][and references therein]{Pit06, Mil18}. We note that the assumption of a symmetric prior on the species frequency distribution does not impose any restrictive constraint on the model. This is evident because, when conducting inference on diversity or, more broadly, on the composition of species assemblies, the specific labels or ordering of species are irrelevant.

Nevertheless, assigning a fixed parameter $\theta$ for the $\textsc{SD}$ prior of the species frequencies, disregarding the species richness $k$, entails severe implications on both prior assumptions and posterior behavior of frequencies and diversity index. Such prior assumption implies that species evenness and the amount of prior information on the species frequencies increase with the number of species $k$. This clearly affects the posterior behavior, since an expected high richness necessarily implies high species evenness and large estimates of the index, regardless the evenness of the observed species frequencies. Hence, we propose a more general class of symmetric Dirichlet priors to avoid this implications, by letting the concentration parameter depends on the number of species $k$. Specifically, $\theta$ is set to be inversely proportional to the number of distinct species $k$. The latter specification, here termed dynamic Dirichlet (\textsc{DD}) distribution, stands out for tractability and interpretability of the corresponding Gini--Simpson index distribution. 

In the following sections, we show that the \textsc{DD} prior is the unique model considered in the manuscript having both prior variability and expected evenness not depending on the number of species $k$. Similarly, as $k$ varies on its support, the amount of prior information on the frequencies $P_j$s is constant under the \textsc{DD} model. The \textsc{DD} prior model is well-behaved for large $k$, as it does not converge to a degenerate distribution when $k$ diverges. In terms of species diversity, the \textsc{DD} prior model allows a separate evaluation of evenness and richness. Hence, the Gini--Simpson index prior expectation factorizes as the product of expected evenness and a simple measure of richness, both controlled by distinct parameters. Remarkably, even from a posterior perspective, the model results to be particularly tractable and interpretable, as the posterior expectation of the Gini--Simpson index is nothing but a convex combination of the minimum variance unbiased estimator and a natural update of the prior guess. 

The second setting we consider along the manuscript is having a prior with an infinite number of species. Hence, the most important prior model known in literature is the one implied by a Dirichlet process \citep[DP,][]{Fer73}. This prior has been investigated in connection with the Gini--Simpson index by~\citet{Arb16}. However, such a prior lacks in flexibility since there is only one parameter to control prior mean and variance of the index, with the prior variance being quite low for large values of the prior mean.  Therefore, we will consider a larger family of priors which include the DP as particular case, namely the Poisson-Dirichlet process \citep[\textsc{PD}, also known as Pitman-Yor process][]{Per92,Pit97}. Such a family still belong to the class of Gibbs-type priors \citep{Gne06}, and enriches the DP as it is indexed by two parameters, both affecting prior mean and variance. \citet{Cer12, Cer14} considers the \textsc{PD} prior case, providing general formulas of the prior and posterior moments of the Gini-Simpson index. However, the inferential relevance of such formulas has not been previously explored in the literature. Here, we derive simple expression of prior and posterior formulas for the \textsc{PD} case, allowing to characterize the Gini--Simpson index distribution trough its first two prior moments. In particular, it turns out that the prior variability for the \textsc{PD} model is increasing in one of its parameters, with the \textsc{DP} model being the uniformly lower variance case. Further, the posterior mean of the index for the PD case can be expressed as a convex combination of the unbiased minimum variance estimator and the prior guess, similarly to the \textsc{DD} model. 

The paper presents also an asymptotic analysis of the Gini--Simpson index with different prior choices. In general, it is nontrivial to study consistency results in a Bayesian nonparametric setting \citep[see, e.g.,][]{Wal01,Wal04,Deb13}. We show that all the models produce consistent estimators of the Gini--Simpson index and have the same asymptotic index posterior, despite their rather different finite sample size behavior. This is quite remarkable, since the DP and PD priors induce a misspecified model if the true number of species is finite. Moreover, we provide a comparison with frequentist asymptotic results, showing that both approaches lead to the same asymptotic uncertainty quantification.

As stated previously, the main contributions of this paper regard the prior and posterior characterization of the Gini–Simpson index, as well as its asymptotic properties. The principal novelties of the manuscript, which are developed in the subsequent sections, can be summarized as follows.
\begin{itemize}
	\item[(a)] We propose a novel prior specification for species diversity estimation based on a symmetric Dirichlet distribution, whose concentration parameter depends on species richness. This formulation induces a well-behaved prior distribution for the Gini–Simpson index and offers advantages in terms of both theoretical properties and posterior inference.
	\item[(b)] We provide a comprehensive prior characterization of the Gini–Simpson index under different prior models, with particular emphasis on prior expectation and dispersion. In addition, we provide a strategy for posterior inference, deriving explicit expressions for posterior point estimates of the index under different prior assumptions.
	\item[(c)] We present an asymptotic analysis for the Gini--Simpson under different prior specifications, highlighting its connection to frequentist estimators. The asymptotic characterization allows the construction of asymptotic credible intervals and thus provides an approach for quantifying posterior uncertainty associated with the index.
\end{itemize}

\begin{figure}[!ht]
\centering
\resizebox{0.99\textwidth}{!}{%
\begin{circuitikz}
\tikzstyle{every node}=[font=\normalsize]
\node [font=\normalsize] at (0,13) {$\textsc{PD}(\sigma, \theta)$};
\draw [->, >=Stealth] (1,13) -- (3.25,13)node[pos=0.5, fill=white]{$\sigma = 0$};
\draw [->, >=Stealth] (12.25,13.5) to [bend right=30] (4.75,13);
\node [font=\normalsize] at (8.75,14.6) {$k = +\infty$};

\node [font=\normalsize] at (4,13) {$\textsc{DP}(\theta)$};
\node [font=\normalsize] at (8,13) {$\textsc{DDM}(\theta, \pi)$};
\node [font=\normalsize] at (8,10.5) {$\textsc{SDM}(\theta, \pi)$};
\draw [->, >=Stealth] (8.6,11) -- (8.6,12.5)node[pos=0.5, fill=white]{$\theta = \frac{\alpha}{K}$};
\draw [->, >=Stealth] (7.4,12.5) -- (7.4,11)node[pos=0.5, fill=white]{$\theta = \alpha$};
\node [font=\normalsize] at (0,10.5) {$\textsc{Gibbs-Type}(\sigma)$};
\draw [->, >=Stealth] (0,11) -- (0,12.75)node[pos=0.5, fill=white]{$\sigma \in [0,1)$, $\textsc{GEM}(\sigma, \theta)$ weights};
\draw [->, >=Stealth] (1.5,10.5) -- (7,10.5)node[pos=0.5, fill=white]{$\sigma < 0$};
\node [font=\normalsize] at (13,13) {$\textsc{DD}(\theta)$};
\node [font=\normalsize] at (13,10.5) {$\textsc{SD}(\theta)$};
\draw [->, >=Stealth] (13.6,11) -- (13.6,12.5)node[pos=0.5, fill=white]{$\theta = \frac{\alpha}{k}$};
\draw [->, >=Stealth] (12.4,12.5) -- (12.4,11)node[pos=0.5, fill=white]{$\theta = \alpha$};
\draw [->, >=Stealth] (9,13.25) -- (12,13.25)node[pos=0.5, fill=white]{$k$ fixed};
\draw [->, >=Stealth] (12,12.75) -- (9,12.75)node[pos=0.5, fill=white]{$K\sim\pi(k)$};
\draw [->, >=Stealth] (9,10.75) -- (12,10.75)node[pos=0.5, fill=white]{$k$ fixed};
\draw [->, >=Stealth] (12,10.25) -- (9,10.25)node[pos=0.5, fill=white]{$K\sim\pi(k)$};
\end{circuitikz}
}%
\caption{Graphical representation of relations among the distinct models considered in the manuscript. Specifically, the Poisson-Dirichlet process (\textsc{PD}), the Dirichlet process (\textsc{DP}), the symmetric Dirichlet (\textsc{SD}) and the dynamic Dirichlet (\textsc{DD}), along with their mixture version (\textsc{SDM} and \textsc{DDM}).}
\label{fig:model_rel}
\end{figure}
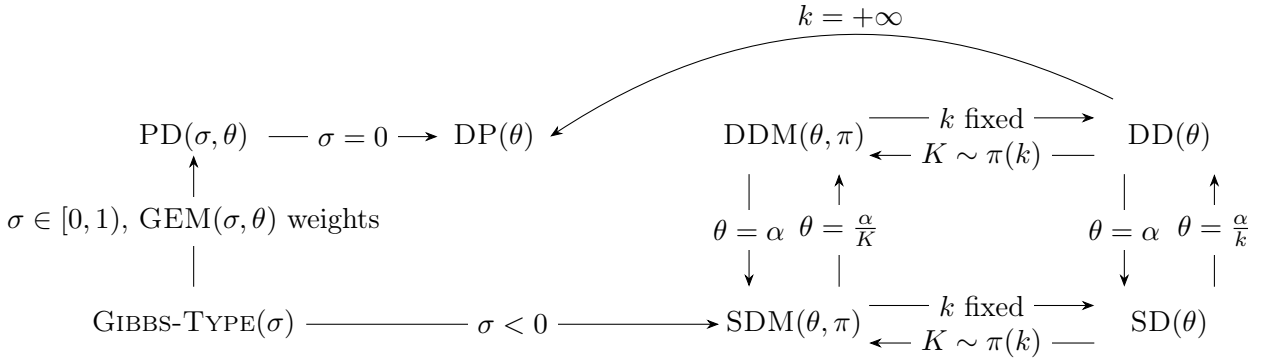

We remark that the models described in this article are related to each others. Figure~\ref{fig:model_rel} shows a visual representation of their connections. Therefore, the results presented in the manuscript are themselves often connected. The rest of the paper is organized as follows. Section~\ref{sec:prel} introduces some preliminary concepts and the prior distributions used along the manuscript. Section~\ref{sec:priorGS} provides a characterization a priori of the Gini--Simpson index with different prior choices. In Section~\ref{sec:postGS} we discuss posterior estimators for the index, with different prior specifications. In Section~\ref{sec:asym} we present the asymptotic distribution of the index under different prior settings. Section~\ref{sec:data} illustrates the index estimation through a species diversity example. The paper concludes with a discussion section. Further examples, illustrations, and proofs of the main results are deferred to the appendix. 

\section{Prior models}\label{sec:prel}

The main object to model a species distribution is the sequence of species frequencies of a population $\{P_j\}_{j = 1}^k$, where $k$ can be fixed or random, possibly infinite. Within a Bayesian approach, the elements $P_j$s of such a sequence are nothing but non-negative contrained random variables. In a general form, following \citet{Pit96}, a species sampling model is a sampling distribution as
\[
\tilde P(\cdot) =  \sum_{j=1}^k P_j \delta_{x_j} + \left( 1 - \sum_{j=1}^k P_j \delta_{x_j}(\cdot) \right) \textsc{G}_0(\cdot), 
\]
where $P_j \geq 0$, $k \in \mathbb N \cup \{+\infty\}$, the generic $x_j \in \mathbb X$ Polish space, and $\sum_{j=1}^k P_j \leq 1$ almost surely. Generally speaking, the sequence of species labels $\{x_j\}_{j=1}^k$ is i.i.d. as $\textsc{G}_0$, but in the context of this paper the $x_j$s only need to be distinct almost surely to function as species labels, since they are not necessarily random. Here, we consider the case for which  $\{P_j\}_{j = 1}^k \in \bigtriangleup^{k-1}$, so that the diffuse term $\textsc{G}_0$ is neglected, hence considering a \textit{proper} species sampling model. Since $\{P_j\}_{j = 1}^k$ and $\{x_j\}_{j = 1}^k$ are random quantities, $\tilde P\sim \mathscr Q$ is also random itself, where $\mathscr Q$ denotes the prior distribution we are considering for $\tilde P$. 

In the following section, we want to characterize the prior and posterior behavior of the Gini--Simpson index $S$ for specific choices of $\mathscr Q$, 
providing explicit expressions for its main prior and posterior summaries. Since the index takes values $[0,1]$, we resort to the normalized variance to clearly understand its behavior. For a generic random variable $W$ taking values in $[0,1]$, the variance value is strongly dependent on its expectation, being close to zero when $\mathbb E[W]$ approaches its boundary values. Specifically, the variance maximum for a given expected value $\mathbb E[W]$ corresponds to $\mathbb E[W](1 - \mathbb E[W])$. Hence, it seems appropriate to quantify the dispersion of $W$ adjusting the variance by its maximum, 
%
\[
	\overline{\textsc{Var}}(W) = \frac{\textsc{Var}(W)}{\mathbb E[W](1 - \mathbb E[W])}.
\]
The previous expression is used consistently in the following sections, to understand and interpret the behavior of several quantities constrained to $[0,1]$. 

\subsection{Prior distributions with finite number of species}

Consider the case with the species richness $K \sim \pi(k)$ being finite almost surely, where $\pi(k)$ denotes its prior, having support on the set of positive integers $\mathbb Z_+$.  Moreover, for each $k$, let $(P_{1,k}, \dots, P_{k,k})$ be a random vector whose distribution is the conditional distribution of $P_1, \dots, P_K \mid K = k$. A common choice, mainly for its tractability, is to set $(P_{1,k}, \dots, P_{k,k}) \sim \textsc{Dir}(\alpha_k)$, where $\textsc{Dir}(\alpha_k)$ denotes the symmetric Dirichlet distribution with concentration parameter $\alpha_k$. Usually, the concentration parameter is fixed to a value that does not depend on the richness, $\alpha_k = \theta$. This model, also known as Fisher model \citep{Fis30}, and hereafter termed  \textsc{SD} model, is also recover by considering a Gibbs--type prior with discount parameter $\sigma < 0$, whereas a Gibbs--type prior with positive discount parameter implies an infinite number of species \citep{Lij08, Pit06}.   

The \textsc{SD} prior model assumption for species frequencies has strong implications on the Gini--Simpson index behavior. Hence, in the following we introduce a prior model having the same frequencies distribution, but with the concentration parameter that now depends of the species richness $k$. Specifically, the concentration parameter is inversely proportional to the richness, with  $(P_{1,k}, \dots, P_{k,k}) \sim \textsc{Dir}(\alpha_k)$ with $\alpha_k = \theta / k$. To distinguish this prior model from the previous symmetric specification, we denote the latter as \textsc{DD} model.  This specification of a symmetric Dirichlet model has been previously considered in literature, but not for species diversity studies. Hence, it has been used to approximate the Ferguson Dirichlet process model \citep[e.g.][]{Kin75, Mul96, Ish02b, Pit06}, since the weight distribution of the \textsc{DD} weakly converges to the one of the DP as $k$ is known and diverges to infinity, as shown by \citet{Mul03}. Further, \citet{Ish02a} consider the \textsc{DD} model as weights prior in a mixture model setting. Later, the same model specification has been considered in the mixture of mixtures framework, resorting to the name dynamical mixture of finite mixtures, which we borrow in this manuscript. See, e.g., \citet{Mcc08,Fru19,Fru21,Gru25}. 

We remark that, even if the \textsc{DD} and  \textsc{SD} models are similar to each other, they induce a quite different prior on the species weights and hence on the Gini--Simpson index. In terms of expectation, both the priors are centering the species weights $(P_{1,k}, \dots, P_{k,k})$ on $1 / k$, since they are symmetric. However, in terms of dispersion, the two priors differ from each other, as the dispersion is controlled by either $\alpha_k$ or $\theta$. The normalized variances of the species weights equal $1 / (\theta + 1)$ for the \textsc{DD} model, independent of the species richness $k$, while for the \textsc{SD} model equals $1 / (k \theta + 1)$, and approaches zero as $k$ diverges to infinity. As consequence, the frequencies of \textsc{SD} model shrink toward $1 / k$ as $k$ increases, having severe impact on the Gini--Simpson index behavior. Further, the amount of prior information entailed in the model corresponds to $\theta$ for the \textsc{DD} model and $k \theta$ for the \textsc{SD} model, which can be interpreted as an equivalent sample size of our prior information. Hence, with the \textsc{SD} model we have a peculiar growth of the prior information strength as the richness increases, while for the \textsc{DD} model the amount of prior information is independent of $k$. This will obviously affect, as we discuss later, posterior estimates of the Gini--Simpson index.

In the following sections, we often consider the more general case where the species frequencies distribution is mixed by a prior assumption for the species richness, hence by considering 
\[
	\begin{split}
	P_1, \dots, P_K \mid K = k &\sim \textsc{DD}(\theta) \quad \text{or} \quad \textsc{SD}(\theta), \\
	K &\sim \pi(k).
	\end{split}
\]
As aforementioned, we refer to these models as \textsc{DDM} and \textsc{SDM}, depending on which model we are mixing. This setting results in a weaker prior specification, since the number of species is not fixed a priori. Specific choices for $\pi(k)$ and their impact on the Gini--Simpson index are discussed in  Section~\ref{sec:prior_rich}.

\subsection{Prior distributions with infinite number of species}\label{sec:PDmodel}

When the number of species is infinite, a suitable prior distribution consists of the Poisson--Dirichlet process. Such a model, also known in literature as Pitman--Yor process \citep{Pit96,Pit97}, extends the celebrated Ferguson DP \citep{Fer73} with an additional parameter $\sigma$, termed discount parameter, that impacts on the frequencies behavior. To have infinitely many positive frequencies $\{P_j\}_{j\geq 1}$, the PD parameters $(\theta, \sigma)$ must satisfy $\sigma \in [0,1)$ and $\theta > -\sigma$. In particular, for $\sigma = 0$ we recover the Ferguson DP with total mass $\theta$. We recall that the frequencies order in the distribution is irrelevant, since the model is invariant with respect to species permutation. Hence, it only matters the probability mass allocation over the frequencies. In the probability literature, the distribution of the process frequencies in ranked order is called Poisson--Dirichlet distribution \citep{Pit96}, but is rater intricate. A more intuitive representation is the size biased permutation $\{\tilde P_j\}_{j\geq 1}$ of $\{P_j\}_{j\geq 1}$, that is the sequence of species probabilities in order of appearance. Such a distribution is also known as the two--parameters Griffiths–Engen–McCloskey distribution \citep[see, e.g.,][]{Pit06}, and admits the following stick--breaking representation 
\[
	\tilde P_1 = W_1, \qquad \tilde P_j = W_j \prod_{\ell = 1}^{j-1} (1 - W_\ell), \; j = 2,  3, \dots, 
\]
where $W_j \sim \textsc{Beta}(1-\sigma, \theta + j \sigma)$, for every $j \geq 1$. 

Both the parameters $\theta$ and $\sigma$ affect the behavior of the species frequencies. Specifically, the strength parameter $\theta$ controls the variability of the frequencies, while the discount parameter $\sigma$ acts on the decay rate, as $j$ growth, and on the asymptotic distribution of the observed species richness, as the sample size $n$ diverges. Indeed, by denoting with $\{P_{(j)}\}_{j\geq 1}$ the decreasing rearrangement of the species frequencies, for the frequencies limit behavior, as $j$ diverges, 
\[
	P_{(j)} \asymp j^{-1 / \sigma}, \text{ if }\sigma > 0, \qquad \qquad P_{(j)} \asymp \mathrm e^{-j / \theta},  \text{ if } \sigma = 0,
\]
where $b_n \asymp c_n$ denotes asymptotic equivalence of $b_n$ and $c_n$, i.e. $\lim_{n \to \infty} a_n / c_n \to \zeta < +\infty$. Hence, the PD model has a power decay of the species frequencies, while the DP shows a faster exponential decay. For the species richness, as the sample size $n$ diverges, we have
\[
	K_n \asymp n^\sigma,  \text{ if } \sigma > 0, \qquad K_n \asymp \log(n),  \text{ if } \sigma = 0, 
\]
where $K_n$ denotes the number of observed distinct species in a sample of size $n$. See, e.g., \citet{Pit06} for further details.  Hence, the PD model mitigates the rich-get-richer property of the DP, where the latter is known to concentrate most of the probability mass over few species. 

\subsection{Differences among prior specifications}

Clearly, assuming a prior distribution with a finite or infinite number of species, and possibly with different frequencies behavior, has an impact on the species sampling process. Table~\ref{tab:summaries} displays the main summaries of the aforementioned species frequencies distributions and the corresponding predictive probability of sampling a new species, given an observed sample. 

\begin{table}[!h]
\centering
\begin{tabular}{l|lll}
                & $\mathbb E[\tilde P_j]$                          & $\overline{\textsc{Var}}(\tilde P_j)$                                                                                            & $\mathrm{Pr}(K_{n+1} = k_{n} + 1 \mid X_1, \dots, X_{n})$ \\ \hline
$\textsc{DD}(\theta)$  & $\frac{1}{k}$                                    & $\frac{1}{\theta + 1}$                                                                                                           &         $\frac{k - k_n}{\theta + n}$                      \\[5pt]
$\textsc{DDM}(\theta, \pi)$ & $\mathbb E \left[\frac{1}{K}\right]$                           
& $\frac{1 + \mathbb E[1 / K]}{(1 + \theta) ( 1 - \mathbb E[1/K])} + \overline{\textsc{Var}}(1 / K)$
&         $\frac{\mathbb E[K \mid k_n] - k_n}{\theta + n}$                    \\[5pt]\hline
$\textsc{SD}(\theta)$   & $\frac{1}{k}$              & $\frac{1}{k \theta + 1}$                                                                                                         &       $\frac{k - k_n}{k \theta + n}$                    \\[5pt]
$\textsc{SDM}(\theta, \pi)$  & $\mathbb E\left[\frac{1}{K}\right]$              
& $\frac{ \mathbb E\left[\frac{K - 1}{K^2 (K \theta + 1)}\right]}{\mathbb E[1 / K]( 1 - \mathbb E[1 / K])} + \overline{\textsc{Var}}(1 / K)$  
&     $\mathbb E\left[\frac{K - k_n}{K \theta + n} \mid k_n\right]$                          \\[5pt]\hline
$\textsc{DP}(\theta)$   & $\frac{1}{\theta + 1}$                           &        $\frac{\theta}{(\theta + 1)^{2j} (\theta + 2)^j}$                &               $\frac{\theta}{\theta + n}$                                            \\[5pt]
$\textsc{PD}(\theta, \sigma)$   & $\frac{1 - \sigma}{\theta + 1 + (j - 1) \sigma}$ &     $\frac{(1 - \sigma)^j(\theta / \sigma + 1)_{j-1}}{ \sigma^{2 j} ( (\theta + 1) / \sigma)_{j}^2 ((\theta + 2) / \sigma)_j }$           &       $\frac{\theta + k_n \sigma}{\theta + n}$                                                       \\[5pt] \hline
\end{tabular}
\caption{Summaries of different prior specifications, specifically the species frequencies expectation, the species frequencies normalized variances and the probability of sampling a new species at the $n+1$ sampling step, given a sample of size $n$. For the DP and PD models, the previous refer to the size-biased representation. We use the Pochammer notation, with $(a)_b = \Gamma(a + b) / \Gamma(b)$, in our framework with $a, b > 0$.}\label{tab:summaries}
\end{table}

From Table~\ref{tab:summaries},  \textsc{DD} and \textsc{DDM} models seem more interpretable of their counterparts \textsc{SD} and \textsc{SDM}. In particular, the normalized variance of the \textsc{DDM} model can be written as function of the first two moments of the richness reciprocal $1 / K$, while the \textsc{SDM} shows a more intricate expression. Given a sample of size $n$, the probability of sampling a new species at the $n+1$ sampling step depends on the sample size $n$ and the number of already observed species $k_n$. However, such a probability is proportional to the richness minus the number of observed species in the \textsc{DD} and \textsc{DDM} models, whereas for the \textsc{SD} and \textsc{SDM} models the number of species appears also in the denominator. These structural differences impact the induced distribution of the Gini-Simpson index and its posterior inference, as it will be show in the next sections. 

\section{Induced prior for the Gini--Simpson index}\label{sec:priorGS}

Different choices for the species frequencies distributions induce different priors for the Gini--Simpson index of equation~\eqref{eq:gsind}. Here we discuss how different models behave, by considering first when the richness is finite but eventually random, and later when the number of species is infinite. We recall that the Gini--Simpson index takes values in $[0,1]$, being equal to zero if there is only a single species and approaching one as the species frequencies tend to be equal and the richness diverges to infinity. In the following, we denote by $S_{\textsc{A}}$ the Gini--Simpson index under the generic prior $\textsc{A}$.

\subsection{Gini--Simpson index with a fixed number of species}

We first consider a generic symmetric Dirichlet models \textsc{SD} and \textsc{DD}, assuming the number of species $k$ to be known and fixed. Some algebraic steps yield to the prior mean and variance of the Gini--Simpson index,
\begin{equation}\label{eq:priDexp}
	\mathbb E[S_{\textsc{SD}}] = \frac{(k - 1) \theta}{ k \theta + 1}, \qquad \qquad \mathbb E[S_{\textsc{DD}}] = \frac{\theta}{\theta + 1} \left( 1 - \frac{1}{k} \right)
\end{equation}
The previous expected values, under both the models \textsc{SD} and \textsc{DD}, increases when the species richness is also increasing. However, for the \textsc{SD} model the previous expectation tends to one as $k$ diverges, while for the \textsc{DD} model tends to $\theta / (1 + \theta)$. Hence, as mentioned in the introduction, the \textsc{SD} model forces the species frequencies to shrink toward $1 / k$ when the number of species is large. Similarly, for the normalized variances we have 
\begin{equation}\label{eq:priDvar}
	\overline{\textsc{Var}} (S_{\textsc{SD}}) = \frac{2}{ (k \theta + 2)(k \theta + 3)}, \qquad \qquad
	\overline{\textsc{Var}} (S_{\textsc{DD}}) = \frac{2}{ (\theta + 2)(\theta + 3)}.
\end{equation}
For the \textsc{SD} model, the variance of the Gini--Simpson index is decreasing in $k$, while is constant for the $\textsc{DD}$ case.

To understand the index behavior, we can consider its normalize version, which is capturing the corresponding evenness of the species. Hence, the Gini--Simpson index reach its maximum value in the equiprobability case, for which equals $1 - 1 / k$, and the normalized Gini--Simpson index equals
\begin{equation*}\label{eq:normind}
	\overline{S} = \frac{S}{1 - 1 / k}. 
\end{equation*}
From the expression of the expectations in equations~\eqref{eq:priDexp} we obtain the following expression for the expectations of the normalized Gini--Simpson index
\begin{equation}\label{eq:normgsind}
\mathbb E\left[\overline{S}_{\textsc{SD}}\right] = \frac{k\theta}{k\theta + 1}, \qquad \qquad \mathbb E\left[\overline{S}_{\textsc{DD}}\right] = \frac{\theta}{\theta + 1}. 
\end{equation}
For the \textsc{SD} model, the prior expected evenness is determined by $k\theta$ and increases with the species richness, while for the \textsc{DD} model the prior expected evenness is constant. Moreover, the \textsc{DD} model clearly separates the expected evenness and richness. Indeed, the prior expectation of the Gini--Simpson index factorizes as the product of expected evenness $E\left[\overline{S}_{\textsc{DD}}\right]$, which depends solely on $\theta$, and a measure of richness $1 - 1/k$. 

The following result consolidates and extend the previous comments regarding the expectation and the variance of the Gini--Simpson index of  Equations~\eqref{eq:priDexp}--\eqref{eq:priDvar} from an asymptotic perspective. 
\begin{proposition}
For the Gini--Simpson index, as the species richness $k$ diverges to infinity, the following results hold true. 
\begin{itemize}
\item[i)] For the \textsc{SD} model, 
\[
S_{\textsc{SD}} \stackrel{L^2}{\longrightarrow} 1, \qquad\qquad S_{\textsc{SD}} \stackrel{a.s.}{\longrightarrow} 1.
\]
\item[ii)] For the \textsc{DD} model, 
\[
S_{\textsc{DD}} \stackrel{d}{\longrightarrow} 1 - \sum_{j = 1}^\infty Q_j^2,
\]
where $\{Q_j\}_{j \geq 1}$ are the weights of a $DP(\theta)$.
\item[iii)] The generic Gini--Simpson index of a symmetric Dirichlet prior distribution with concentration parameter $\beta_k$ converges in distribution to a random variable that does not degenerate at zero or one if and only if its total concentration $k \beta_k$ converges to a constant that is neither zero or infinity. 
\end{itemize}
\end{proposition}
Thus, the \textsc{DD} model, with $k \beta_k \to \theta$ as $k$ diverges, is essentially the only model within the possible symmetric Dirichlet specifications that allows for a non-degenerate Gini--Simpson index, as the richness diverges. Therefore, all other specifications imply extreme types of populations in the limit behavior, either with only one species or with infinitely many having all similar frequencies. 

A crucial aspect regards the uncertainty quantification associated with the distribution induced by different prior specifications. Hence, for the \textsc{SD} and \textsc{DD} models, we are able to provide explicitly an expression of the minimum prior variance, when fixing the index expectation to a specific value $\mu$. 

\begin{proposition}
	Consider the Gini--Simpson index with fixed richness and let $\mu = \mathbb E[S]$ be the prior expectation of the index. For both SD and DD models, the minimum of the normalized variance equals
		\[
			\min_{k \in \mathbb Z_+}\left\{\overline{\textsc{Var}}(S_\textsc{SD})  \right\} = \min_{k \in \mathbb Z_+}\left\{\overline{\textsc{Var}}(S_\textsc{DD})  \right\} = \frac{2\left[(1-\mu)g(\mu)-1\right]^2}{\left[(2-\mu)g(\mu)-2\right]\left[(3-2\mu)g(\mu)-3\right]},
		\]
		where $g(\mu)= \lceil 1/(1-\mu) \rceil$ corresponds to the optimal value of the richness and $ \lceil \cdot \rceil$ denotes the ceiling function.
\end{proposition}

\subsection{Gini--Simpson index with a random number of species}

As mentioned before, we can relax the model specifications by considering the number of species $K$ unknown and random, with $K \sim \pi(k)$. Hence, the prior means for the \textsc{SDM} and \textsc{DDM} models equal
\begin{equation}\label{eq:exp_prior_k}
\mathbb E\left[{S}_{\textsc{SDM}}\right] = \mathbb E\left[ \frac{(K - 1) \theta}{K \theta + 1} \right], \qquad \qquad \mathbb E\left[{S}_{\textsc{DDM}}\right] = \frac{\theta}{\theta + 1} (1 - \mathbb E[1 / K]), 
\end{equation}
where the lefthand expectations in the previous equations are taken with respect to $K$. Both the previous expression are non-decreasing function of $\theta$, and they range in the interval $(0, 1 - \mathbb E[1 / K])$. Similarly, the prior expectations of the normalized Gini--Simpson index can be retrieved evaluating the expectations, with respect to the species richness, of the expression stated in equation~\eqref{eq:normgsind}, obtaining 
\[
\mathbb E\left[\overline{S}_{\textsc{SDM}}\right] = \mathbb E \left[\frac{K\theta}{K\theta + 1}\right], \qquad \qquad \mathbb E\left[\overline{S}_{\textsc{DDM}}\right] = \frac{\theta}{\theta + 1}. 
\]
We can appreciate that letting the richness being a random quantity does not affect the prior expectation of the normalized index for the \textsc{DDM}. Further, for the \textsc{DDM} model holds a factorization of the prior expectation similar to the \textsc{DD} case, for which we can decompose the unnormalized index in two terms, a first term measuring the species evenness $\theta/ (\theta +1 )$ and a second term measuring the species richness $1 - \mathbb E[1 / K]$.

Regarding the normalized variances of the Gini--Simpson index for the mixture cases, we do not report here the expression for the \textsc{SDM} case, since is quite cumbersome. For the \textsc{DDM} case, after some algebraic manipulation the normalized variance can be expressed as
\begin{equation}\label{eq:var_prior_k}
\overline{\textsc{Var}} (S_{\textsc{DDM}}) = 1 - \left(1 - \frac{2}{(\theta + 2) (\theta + 3)} \right)\left(1 - \overline{\textsc{Var}} (1 / K) \frac{\theta \mathbb E[1 / K]}{1 + \theta \mathbb E[1 / K]} \right).
\end{equation}
The previous expression is an increasing function of the normalized variance of $1 / K$, which stands for the richness variance when we decompose $S_{\textsc{DDM}}$ in the product of two terms. In particular, when $K$ degenerates,  $\overline{\textsc{Var}} (S_{\textsc{DDM}})$ reaches its minimum value, which corresponds to  $\overline{\textsc{Var}} (S_{\textsc{DD}})$. Further, both prior mean and normalized variance of the \textsc{DDM} Gini--Simpson index depend on the richness prior distribution $\pi$ only through the first two moments of $1 / K$. Therefore, they are entirely determined in the \textsc{DDM} model by three parameters, namely $\theta$, $\mathbb E[1 / K]$ and $\overline{\textsc{Var}}(1 / K)$, which all are easily interpretable. 

To highlight the flexibility of both \textsc{SDM} and \textsc{DDM} models in describing prior opinions on the index values, we can study the range of its first two moments for different prior choices. The following proposition summarizes the expectation support and the supremum of the Gini--Simpson index normalized variance for both models, while the prior mean is held fixed. 
\begin{proposition}\label{prop:supre}
The following conditions for the symmetric Dirichlet mixture priors hold true.
\begin{itemize}
\item[i)] For both \textsc{SDM} and \textsc{DDM} models, the prior expectation of the Gini--Simpson index takes values in $(0, 1 - \mathbb E[1 / K])$.
\item[ii)] For both \textsc{SDM} and \textsc{DDM} models, and for any given value $\mu$ of the Gini--Simpson index prior expectation, the supremum of the normalized variance of the index equals one, i.e.
\[
	\sup_{\pi, \theta} \left\{ \overline{\textsc{Var}}(S) \right\} = 1,
\]
as the prior $\pi$ of the species richness and the concentration parameter $\theta$ vary. 
\end{itemize}
\end{proposition}
Hence, for both model \textsc{SDM} and \textsc{DDM} the prior variability can be chosen arbitrarily large, for any given prior of the mean.  Specifically, $\overline{\textsc{Var}}(S)$ approaches its supremum if the normalized variance of the species richness $K$ is close to one, which is tantamount to say that $K$ is either one or infinite almost surely. There is no explicit expression of the variance minimums for \textsc{SDM} and \textsc{DDM}, as for the fixed number of species case. However, for both models we established an optimization problem to obtain a numerical evaluation for the variance minimums, regardless the choice of the prior $\pi(k)$ for the species richness. Details on the derivation of such minimums are deferred to the appendix, while a graphical illustration is provided in Figure~\ref{fig:EVAR}. 


\subsection{Gini--Simpson index with an infinite number of species}

We consider now the \textsc{PD} model introduced in Section~\ref{sec:PDmodel}, where the number of species is assumed to be infinite. Hence, for the corresponding Gini--Simpson index $S_{\textsc{PD}}$, by some algebraic steps, one obtains the following expressions for its expectation and normalized variance 
\begin{equation}\label{eq:priorPD}
\mathbb E[S_{\textsc{PD}}] = \frac{\theta + \sigma}{\theta + 1}, \qquad\qquad \overline{\textsc{Var}}(S_{\textsc{PD}}) = \frac{2}{(\theta+2)(\theta+3)}, 
\end{equation}
where $\theta > -\sigma$ and $\sigma \in [0,1)$. Equation~\ref{eq:priorPD} indicates that the prior mean increases monotonically with both $\theta$ and $\sigma$, taking values in the whole unit interval. Conversely, the normalized variance depends exclusively on the strength parameter and decreases as $\theta$ increases. The latter can take any value in the unit interval, but it is constrained once the prior mean is fixed, as shown in the following proposition. 
\begin{proposition}\label{prop:varPD}
For any given value of the prior mean $\mu$, for $S_{\textsc{PD}}$, i.e. the Gini--Simpson index under the Poisson--Dirichlet process prior, we have
\[
	\overline{\textsc{Var}}(S_{\textsc{PD}}) = \frac{2(1 - \mu)^2}{(2 - \mu - \sigma)(3 - 2\mu - \sigma)}. 
\]
\end{proposition}
We can clearly state from Proposition~\ref{prop:varPD} that $\overline{\textsc{Var}}(S_{\textsc{PD}}) $ is an increasing function of $\sigma$, for any given $\mu$. Therefore, the index prior variability entailed by the \textsc{PD} prior model has a lower bound, corresponding to the \textsc{DP} case ($\sigma = 0$), while no upper constraints are present, and the maximum variability tends to one as $\sigma \to 1$. 

\subsection{Comparing different prior specifications for the Gini--Simpson index}

Regarding the index expectation with different prior choices, a key distinction lies in the infinite richness assumption of the \textsc{PD} model. Consequently, no explicit measure of richness appears in its expected index. 
Furthermore, the expectation of the \textsc{DD} model converges to that of the \textsc{DP} one as the richness $k$ diverges to infinity. This behavior is consistent with the \textsc{DD} converging in distribution to a \textsc{DP}, when the number of species $k$ diverges. From the left panel of Figure~\ref{fig:EVAR}, we can appreciate that the prior expectations of both \textsc{SD} and \textsc{SDM} models, here assuming the richness prior of Section~\ref{sec:prior_rich}, rapidly approach a plateau value of the Gini--Simpson index, while other models are more flexible over the index support. Specifically, \textsc{DD}, \textsc{DDM} and \textsc{DP} models have a similar behavior of the expectation growth as function of $\theta$, while  introducing a discount parameter for the \textsc{PD} model changes the expectation of the index for small values of $\theta$, but preserve a similar behavior as $\theta$ increases. 

\begin{figure}[h]
	\includegraphics[width = 0.99\textwidth]{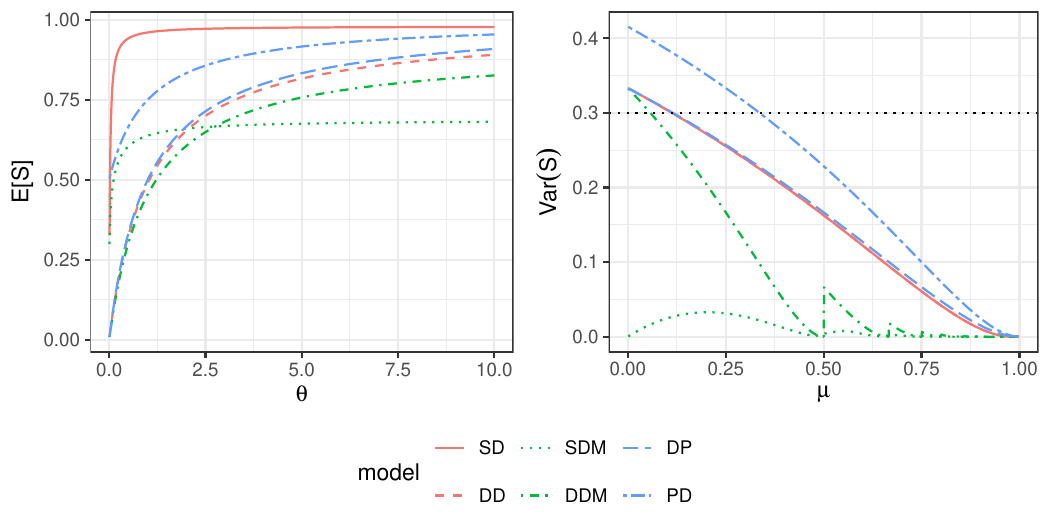}
	\caption{Left panel: prior expectation of the Gini--Simpson index as function of $\theta$, \textsc{SDM} and \textsc{DDM} with the priors of section~\ref{sec:prior_rich}, setting $(\psi_1, \psi_2) = (0.9, 0.9)$ for the SDM richness prior of equation~\eqref{eq:prior_rich_sdm} and $(\gamma_1, \gamma_2) = (10,10)$ for the DDM richness prior of equation~\eqref{eq:prior_ksdm}. Right panel: minimum of the prior normalized variance for the Gini--Simpson index distribution, as function of its prior expectation. \textsc{SD} and \textsc{DD} model are specified with $k = 50$ and the \textsc{PD} model with $\sigma = 0.25$. Red lines correspond to a finite and fixed number of species, green lines to a finite but random number of species, blue lines to an infinite number of species.}\label{fig:EVAR}
\end{figure}

For the normalized variances, all three prior models may reach maximum index variability for any given prior mean. Hence, to compare their flexibility in terms of support of the prior dispersion of the Gini--Simpson index distribution, it is sufficient to evaluate the minimum of the normalized prior variance as function of the prior mean $\mu$. The right panel of Figure~\ref{fig:EVAR} shows the curves associated to different models. We first remark that for the \textsc{DP} and \textsc{PD} models, the right panel does not show the actual normalized variance of the models, not the minimums. The normalized prior variance minimums, as function of the prior expectation $\mu$, coincide for the \textsc{SD} and \textsc{DD} models, which are close to the variance arising from the \textsc{DP} model. The variance for the \textsc{PD} model with $\sigma = 0.25$ has a similar decay to the \textsc{DP} prior, as function of the prior expectation $\mu$, but having overall higher values, especially for small values of $\mu$.  The minimum variance of the \textsc{SDM} model is substantially lower than the \textsc{DDM} model minimum when $\mu < 1/2$, otherwise the two models have comparable small minimums. We remark that small values of the index are not frequently encountered in real case studies. If the index is smaller than a given value $c \in (0,1)$, then there is a species whose frequency is bigger than $1 - c$, which rarely occurs for $c < 1/ 2$. Further, the minimum variance curves of \textsc{SDM} and \textsc{DDM} models are globally lower than the other ones, thereby providing greater flexibility in the resulting prior distributions. Figure~\ref{fig:EVAR} highlights a limitation of the \textsc{DP} and \textsc{PD} models. In the region of the index support most frequently encountered in practice, namely for large values of $\mu$, the variance drops rapidly, producing an excessively concentrated prior distribution and consequently limiting prior flexibility.

To have an idea of plausible values of the normalized variance of the Gini--Simpson index, one may consider unimodal priors with large variance, which are often used in Bayesian statistics. For example, we can consider the beta distribution as a benchmark, given its high flexibility since its first two moments can be arbitrary fixed. It is straightforward to verify that the normalized variance of any unimodal beta random variable is smaller than $1/3$, emphasized as horizontal line in the left panel of Figure~\ref{fig:EVAR}. On the contrary, values of the normalized variance close to one are rather extreme and rarely correspond to a genuine prior opinion. In fact, they imply a population which has either one species or infinitely many with similar frequencies. 

Among the models,  \textsc{SD}, \textsc{DD}, \textsc{DP} and \textsc{PD} exhibit some limitations in capturing the variability if the prior expectation is not close to one, whereas similar limitations are shared by the \textsc{DDM} model only if such expectation is close to zero. The \textsc{DP} appears to be inadequate as a general prior for the Gini--Simpson index, since we can only fix either the prior expectation or the prior variability of the index. Another interesting subclass of the \textsc{PD} model is obtained by normalizing a stable independent increment process, corresponding to let $\theta = 0$ in the \textsc{PD} prior. In this case, $\sigma$ plays the role of index mean and can be arbitrarly chosen, but the normalized prior variance of the index is constant and equals $1 / 3$, which appears rather restrictive. 

\section{Posterior inference for the Gini--Simpson index}\label{sec:postGS}

Suppose we observe an $\mathbb X$-valued sample of species $X_1, \dots, X_n$ with $k_n$ distinct observed species having frequencies $n_1, \dots, n_{k_n}$. In this section, we summarize the exact forms of posterior point estimates of the Gini--Simpson index with different priors. 
As point estimate, we consider the posterior expectation of the Gini--Simpson index. As it will be shown, these estimates have simple and interpretable form, as they are expressed as convex combination of a frequentist estimator and the prior expectations.

\subsection{Posterior estimates with a fixed number of species}

First, we consider the case with a symmetric Dirichlet prior having fixed and known richness, hence the \textsc{SD} and \textsc{DD} models. The posterior expectation of the Gini--Simpson index can be expressed as 
\begin{equation}\label{eq:postSD}
	\mathbb E[ S_{\textsc{SD}} \mid X_1, \dots, X_n] = p_{k,n}^{\textsc{SD}} \hat S_{n} + \left(1 -  p_{k,n}^{\textsc{SD}}\right)  \frac{(k - 1) \theta}{ k \theta + 1}
\end{equation}
for the \textsc{SD} model, and
\begin{equation}\label{eq:postkSD}
	\mathbb E[ S_{\textsc{DD}} \mid X_1, \dots, X_n] = p_{k,n}^{\textsc{DD}} \hat S_{n} + \left(1 -  p_{k,n}^{\textsc{DD}}\right)  \frac{\theta}{\theta + 1} \left( 1 - \frac{1}{k} \right)
\end{equation}
for the \textsc{DD} case, where 
\[
	p_{k,n}^{\textsc{SD}} = \frac{n(n-1)}{(k \theta + n + 1)(k \theta + n)}, \qquad \qquad
	p_{k,n}^{\textsc{DD}} = \frac{n(n-1)}{(\theta + n + 1)(\theta + n)},
\]
are the weights in Equations~\eqref{eq:postSD}-\eqref{eq:postkSD}, and 
\begin{equation}\label{eq:freq_est}
	\hat S_{n} = 1 - \sum_{j = 1}^{k_n} \frac{n_j (n_j - 1)}{n (n - 1)}
\end{equation}
denotes the unbiased (minimum variance) frequentist estimator of the Gini--Simpson index. We remark that in the right-hand terms of Equations~\eqref{eq:postSD}-\eqref{eq:postkSD} appear the prior expectations of the \textsc{SD} and \textsc{DD} models. Therefore, the posterior point estimates of the Gini--Simpson index are linear convex combinations of the frequentist estimator $\hat S_{n}$ and the prior guess of the index. The weight of the empirical estimator for the \textsc{SD} model depends on the sample size, the number of distinct species observed in the sample and the concentration parameter, while for the \textsc{DD} model such a weight depends solely on the sample size and the concentration parameter. Thus, in the posterior expectations the \textsc{SD} model has an increasing importance of the prior opinion as the richness increases, while for the \textsc{DD} model the prior expectation weight does not depend on the number of species. 

\subsection{Posterior estimates with a random number of species}

Regarding the mixture case, the posterior mean of the index can be obtained mixturing equations~\eqref{eq:postSD}-\eqref{eq:postkSD} with respect to the posterior of the richness $K$, that is given by  
\[
\begin{split}
	\pi_n^{\textsc{SDM}}(k) &\propto \frac{(k - k_n + 1)_{k_n}}{(k \theta)_n} \pi(k) \mathbb I_{[k \geq k_n]}, \\
	\pi_n^{\textsc{DDM}}(k) &\propto \frac{(k - k_n + 1)_{k_n}}{(\theta)_n} \prod_{j = 1}^{k_n} ({\theta} / {k})_{n_j} \pi(k)\mathbb I_{[k \geq k_n]}, 
\end{split}
\]
for the \textsc{SDM} and \textsc{DDM} models, respectively, where $\mathbb I_{[A]}$ denotes the indicator function over the set $A$ and $(a)_b = \Gamma(a + b) / \Gamma(a)$, $b \geq 0$, denotes the rising factorial, as before. A distinctive peculiarity of the \textsc{SDM} model lies in the posterior distribution of the richness, as it does not depend on the observed species frequencies $n_1, \dots, n_{k_n}$ in the sample. Indeed, the \textsc{SDM} belongs to the Gibbs--type priors \citep{Pit06}. Instead, with the \textsc{DDM} model, sample frequencies affect the posterior distribution of the richness. 

In general, for a symmetric Dirichlet prior model, the posterior expectation of the Gini--Simpson index can be expressed as
\begin{align}
		\mathbb E[S \mid X_1, \dots, X_n] &= q_n \hat S_n + (1 - q_n) \sum_{k = k_n}^{\infty} \mathbb E[S \mid K = k] \tau_n(k) \label{eq:symm_est} \\
		&= q_n \hat S_n + (1 - q_n) W_n, \label{eq:symm_est_comp}
\end{align}
where $E[S \mid K = k]$ denotes the prior expectation of either the \textsc{SD} or \textsc{DD} models, $W_n$ depends on the specific prior model, and 
\[
	 q_n = \sum_{\ell = k_n}^{\infty} p_{\ell,n} \pi_n(\ell), \qquad \tau_n(k) = \frac{(1 - p_{k,n}) \pi_n(k)}{\sum_{\ell = k_n}^\infty (1 - p_{\ell,n}) \pi_n(\ell)},
\]
for every $k$. From the previous expression, the posterior mean of the Gini--Simpson index, with a symmetric Dirichlet prior, is also a convex linear combination of the unbiased frequentist estimator $\hat S_n$ and an estimator of such an index which updates the prior guess through the pseudo--posterior $\tau_n$ of $K$. Hence, the estimator in Equation~\eqref{eq:symm_est} is a smoother alternative to the frequentist estimator $\hat S_n$, when the richness is unknown and random. Using the \textsc{SDM} model, the pseudo--posterior distribution $\tau_n(k)$ is obtained from $\pi_n(k)$ by moving mass from lower to higher values, whereas $\tau_n(k)$ and $\pi_n(k)$ coincide if and only if the \textsc{DDM} model is employed. 

The posterior mean of the \textsc{SDM} model does not admit an easily interpretable expression, since the terms in equation~\eqref{eq:symm_est_comp} becomes
\[
\begin{split}
	q_n^{\textsc{SDM}} = n (n - 1) \mathbb E_{\pi_n} \left[ \frac{1}{(\theta K + n)(\theta K + n + 1)} \right], \qquad
	W_n^{\textsc{SDM}} = 1 - \frac{\theta + 1}{n} \left( \frac{n \xi_n + 1}{1 - q_n^{\textsc{SDM}}} - 1 \right), 
\end{split}
\]
with $\xi_n = \mathbb E_{\pi_n}\big[ (\theta K + n)^{-1}\big]$ and $\mathbb E_{\pi_n}[\cdot]$ denoting the expectation with respect to the posterior distribution $\pi_n$ of $K$. The weight $q_n^{\textsc{SDM}}$ approaches its minimum of zero when $\theta$ is large or the posterior distribution $\pi_n$ of $K$ concentrate mass on large values of its support. Thus, for the latter case, the posterior mean of the index approaches $W_n^{\textsc{SDM}}$, which in turn tends to $1$, disregarding the values of the sample frequencies $n_j$s. This also relates to the posterior distribution of $K$, being independent of the sample frequencies. Hence, with the \textsc{SDM} model, high evenness is attained just because high richness is expected. In particular, when $K$ is believed large under the prior distribution, the prior bias toward high values affects also the posterior distribution of the index, no matter what is the real evenness of the observed species frequencies. The behavior is confirmed by asymptotic analysis, since $W_n^{\textsc{SDM}} \asymp 1 - 2 / (k_n + 1)$, which rapidly approaches one as $k_n$ increases, assuming $k_n / n \to 0$. On the counterpart, when $k_n / n \to c$ as $n$ diverges, with $c \in (0,1)$,  $W_n^{\textsc{SDM}}$ tends to one and $q_n^{\textsc{SDM}}$ is upper-bounded by $(\theta c + 1)^{-2}$. Therefore, even in the limit case, the posterior mean of the index in equation~\eqref{eq:symm_est_comp} is systematically larger than the frequentist unbiased estimator $\hat S_n$, disregarding what unbounded prior for $K$ and value of $\theta$ are chosen and species frequencies are observed. 

With the \textsc{DDM} model, the quantities in the index posterior expectation of Equation~\eqref{eq:symm_est_comp} take form 
\begin{equation}\label{eq:pm_ksdm}
	q_n^{\textsc{DDM}} = \frac{n(n-1)}{(\theta + n + 1)(\theta + n)}, \qquad W_n^{\textsc{DDM}} = \frac{\theta}{1 + \theta} \left( 1 - \mathbb E_{\pi_n}\left[ \frac{1}{K}\right]\right). 
\end{equation}
The posterior estimator we obtain is a convex linear combination of the frequentist unbiased estimator $\hat S_n$ and the estimator $W_n^{\textsc{DDM}}$, where the latter can be viewed as a natural updating of the prior estimate $\mathbb E[S_{\textsc{DDM}}] = \frac{\theta}{1+\theta}(1 - \mathbb E[1 / K])$, which incorporates available data information on the richness. In the limit behavior, 
\[
	W_n^{\textsc{DDM}} \asymp \frac{\theta}{1 + \theta}\left(1 - \frac{1}{k_n}\right), 
\]
which serves also as lower bound for $W_n^{\textsc{DDM}}$. This is straightforward since $\pi_n$ is supported on the integers greater or equal than $k_n$, and converges to a point mass on $k_n$ as the sample size $n$ diverges. 

\subsection{Posterior estimates with an infinite number of species}

The posterior distribution of a \textsc{PD} model is well known in literature \citep[][Corollary 20]{Pit96}. Here, it can be exploited to obtain the posterior distribution of the Gini--Simpson index under the \textsc{PD} prior model. In particular, the posterior distribution of the $S_{\textsc{PD}}$ index can be expressed as 
\[
	S_{\textsc{PD}} \mid X_1, \dots, X_n \sim 1 - \eta_n^2 \sum_{\ell \geq 1} Q_\ell^2 - (1 - \eta_n^2) \sum_{i = 1}^{k_n} \xi_{i}^2, 
\]
where $(\xi_1, \dots, \xi_{k_n}) \sim \textsc{Dir}(n_1 - \sigma, \dots, n_{k_n} - \sigma)$, $\eta_n \sim \textsc{Beta}(\theta + k_n \sigma, n - k_n \sigma)$, and $\{ Q_j\}_{j \geq 1}$ are the weights of a \textsc{PD} model with parameters $(\theta + k_n \sigma, \sigma)$, with $(\xi_1, \dots, \xi_{k_n})$, $\eta_n$ and $\{Q_j\}_{k \geq 1}$ being stochastically independent. After some manipulations, the posterior mean of the $S_{\textsc{PD}}$ index can be expressed as a linear convex combination of the prior mean $\mathbb E[S_{\textsc{PD}}]$ and the frequentist unbiased estimator $\hat S_n$, 
\begin{equation}\label{eq:pm_pd}
	\mathbb E[S_{\textsc{PD}} \mid X_1, \dots, X_n] = q_n^{\textsc{PD}} \hat S_n + \left(1 - q_n^{\textsc{PD}}\right) \mathbb E[S_{\textsc{PD}}], 
\end{equation}
where 
\[
	q_n^{\textsc{PD}} = \frac{n(n-1)}{(\theta + n + 1)(\theta + n)}. 
\]
We first remark that $q_n^{\textsc{PD}}$ of the \textsc{PD} weights the empirical information exactly as done for the \textsc{DDM} model in equation~\eqref{eq:pm_ksdm}. Notice that the weight $1 - q_n^{\textsc{PD}}$, given to the prior opinion, is increasing in $\theta$, coherently with the prior variance of th eindex being decreasing in $\theta$. 

Quite remarkably, the estimator expressed in equation~\eqref{eq:pm_pd} is similar to the one provided by model \textsc{DDM}, but now the estimator $W_n^{\textsc{DDM}}$ is replaced by the prior mean. Thus, the \textsc{DDM} posterior estimator incorporates sample information more effectively by accounting for the number of species $K$, since $W_n^{\textsc{DDM}}$ updates the index prior mean considering also the posterior distribution of $K$. On the contrary, the \textsc{PD} model assumes the richness known and equal to infinity, and its posterior estimator does not include sample information on $K$. Interestingly, from equation~\eqref{eq:pm_pd} we can appreciate that the frequentist unbiased estimator $\hat S_n$ is recovered as limit case of the Bayesian estimator with \textsc{PD} prior, for $\sigma \to 1$ and $\theta \to - 1$. In such limit case the normalized prior variance of $S_{\textsc{PD}}$ goes to one, i.e. reaches its maximum, and the limit prior of the index is concentrated on two points, namely the case with only one species or with infinitely many having similar frequencies.

\section{Posterior asymptotics}\label{sec:asym}

An interesting quantity to study is the asymptotic distribution of the Gini--Simpson index induced by different prior models for the species frequencies. Characterizing the asymptotic posterior distributions is fundamental to study the limiting behavior of the index and its uncertainty quantification, comparing Bayesian credible intervals with frequentist confidence intervals. In the following, we derive those distribution under the \textit{what if} principle of Diaconis and Feedman, hence assuming $X_1, \dots, X_n$ i.i.d. discrete random variables. Furthermore, we denote by $\mathbb P_0$ a probability measure that makes observations $X_1, \dots, X_n$ i.i.d., with their common distribution having masses $p_1, \dots, p_{k_0}$, $k_0$ denotes the number of species in the entire population, and $S_0$ being the true corresponding value of the Gini--Simpson index. 

For such an index, the most common frequentist estimators are the aforementioned  unbiased estimator with minimum variance $\hat S_n$ and the plug--in estimator, i.e., 
\[
	\tilde S_n = 1 - \sum_{j = 1}^{k_n} \frac{n_j^2}{n^2}, 
\]
where $k_n$ is again the observed number of species out of a sample $X_1, \dots, X_n$ and $n_1, \dots, n_{k_n}$ the observed species frequencies. We remark that, as a direct consequence of Slustky's Theorem, the asymptotic distributions of $\hat S_n$ and $\tilde S_n$ coincide. In particular, as $n$ diverges, 
\begin{equation}\label{eq:asy_freq}
	\lim_{n\to + \infty} \sqrt{n} (\tilde S_n - S_0) \stackrel{d}{\longrightarrow} \textsc{N}\left(0, \sigma_{n}^2\right), \qquad \text{with} \qquad \sigma_{n}^2 = 4 \left[ \sum_{j=1}^{k_0} p_j^3 - \left( \sum_{j=1}^{k_0} p_j^2 \right)^2 \right],
\end{equation}
where the true species frequencies $p_1, \dots, p_{k_0}$ are assumed not to be all equal. The result in equation~\ref{eq:asy_freq} can be obtained applying the multivariate central limit theorem and the delta method \citep[see, e.g.,][]{Lyo81}. As a consequence, for any given confidence level $(1 - \alpha) \in (0,1)$, the asymptotic confidence interval for the Gini--Simpson index can be expressed as
\begin{equation}\label{eq:interval}
	\left( \tilde S_n - \Phi^{-1} (1 - \alpha / 2) \frac{s_{n}}{\sqrt{n}},\;  \tilde S_n + \Phi^{-1} (1 - \alpha / 2) \frac{s_{n}}{\sqrt{n}}\right), 
\end{equation}
where $\Phi^{-1}$ denotes the quantile function of the standardized normal distribution and $s_{n}^2$ is a suitable estimator of the variance term in equation~\eqref{eq:asy_freq}. For instance, an estimator that is consistent by the strong law of large numbers is given by 
\begin{equation}\label{eq:est_var}
	s^2_{n} = 4 \left[ \sum_{j=1}^{k_0} \frac{n_j^3}{n^3} - \left( \sum_{j=1}^{k_0} \frac{n_j^2}{n^2}\right)^2 \right].
\end{equation}

Once we assume a prior distribution on the species frequencies, and we observe a sample, our information is updated and the way the posterior distribution behaves and concentrates is fundamental for statistical analysis. In the following, we consider all the models presented in the manuscript, but for the case with known richness we assume that it is fixed to a value equal or larger than the population one, i.e. $k \geq k_0$ for \textsc{SD} and \textsc{DD} models.  A crucial distinction to derive their posterior asymptotics regards the population form which the sample comes from, in particular if all species are equally likely or not. We separate those cases in the following proposition. 

\begin{proposition}\label{theo:asym_post}
	Let $S_n$ be a random variable whose distribution corresponds to the posterior of the Gini--Simpson index under either \textsc{SDM}, \textsc{DDM}, \textsc{DP} and \textsc{PD} models or \textsc{SD} and \textsc{DD} models with $k \geq k_0$.
	Then, the following distributional results hold true.
	\begin{itemize}
	\item[(i)] If $p_1,\dots, p_{k_0}$ are not all equal to $k_0^{-1}$ then, $\mathbb P_0$ almost surely,
	\[
		\lim_{n \to +\infty}\sqrt{n}(S_n - \tilde S_n)  \stackrel{d}{\longrightarrow} \textsc{N}\left(0, \lambda_{n}^2\right), \qquad \text{with} \qquad \lambda_n^{2} = 4 \left[ \sum_{j=1}^{k_0} p_j^3 - \left(\sum_{j=1}^{k_0} p_j^2\right)^2 \right].
	\]
	\item[(ii)] If $p_1,\dots, p_{k_0}$ are all equal to $k_0^{-1}$ then,  $\mathbb P_0$ almost surely,
	\[
		\lim_{n \to +\infty} n k_0 (S_n - \tilde S_n) \stackrel{d}{\longrightarrow} \chi^2_{k_0 - 1}, 
	\]
	a Chi-squared distribution with $k_0 - 1$ degrees of freedom. 
	\end{itemize}
\end{proposition}

For all models studied in the manuscript, $n^\nu (S_n - \tilde S)$ converges in distribution, whereas both $\nu$ and the limit distribution depend on a certain condition on the population species frequencies, specifically if the species are equiprobable. By Slustky's Theorem, the previous distributional convergences imply consistency of the posterior distribution, i.e. $S_n$ converges to the true value of the index, and therefore also the posterior expectation of the Gini--Simpson index is a consistent estimator. Interestingly, the previous proposition shows that, in an asymptotic regime, the models have the same behavior, quite remarkably since the \textsc{DP} and \textsc{PD} models are misspecified, as they entail an infinite richness. Regarding the Gini--Simpson index uncertainty quantification, the following corollary provides a strategy to construct asymptotic credible intervals. 

\begin{corollary}\label{cor:norm}
Let $s^2_n$ be a strongly consistent estimator of the asymptotic variance, e.g. the estimator in equation~\eqref{eq:est_var}. Under the same assumptions of (i) in Proposition~\ref{theo:asym_post}, 
\[
\lim_{n \to +\infty} \sqrt{\frac{n}{s^2_n}} (S_n - \tilde S_n)  \stackrel{d}{\longrightarrow} \textsc{N}\left(0, 1\right),
\]
and therefore for every $\alpha \in (0,1)$, the posterior probability that the Gini--Simpson index belongs to the interval of equation~\eqref{eq:interval} converges to $1 - \alpha$, $\mathbb P_0$ almost surely. 
\end{corollary}

Assuming one of the prior models of Proposition~\ref{theo:asym_post}, for a given level $\alpha \in (0,1)$, from Corollary~\ref{cor:norm} the asymptotic credible interval for the Gini--Simpson index coincides with the asymptotic confidence interval. Hence, the uncertainty quantification under a Bayesian setting is equivalent to the frequentist one, as the sample size diverges. 

\section{Analysis of the North America ranidae dataset}\label{sec:data}

To illustrate the use and differences of the previous prior specification of the Gini--Simpson index, we analyze a set of species detection data from the Global Biodiversity Information Facility project \citep{Gbi21}. Such a project is an extensive database consisting in record of species found across the world, where for each individual is reported the taxonomy, location and possibly other relevant information. The sample we consider is a set of data of the Ranidae family observed in North
America, for which the species is identified by the scientific name. In particular, the raw data have been processed, obtaining a total number of $n = 71\,856$ observed individual divided in $k_n = 109$ distinct species. 

\begin{figure}[h]
	\includegraphics[width = 0.95\textwidth]{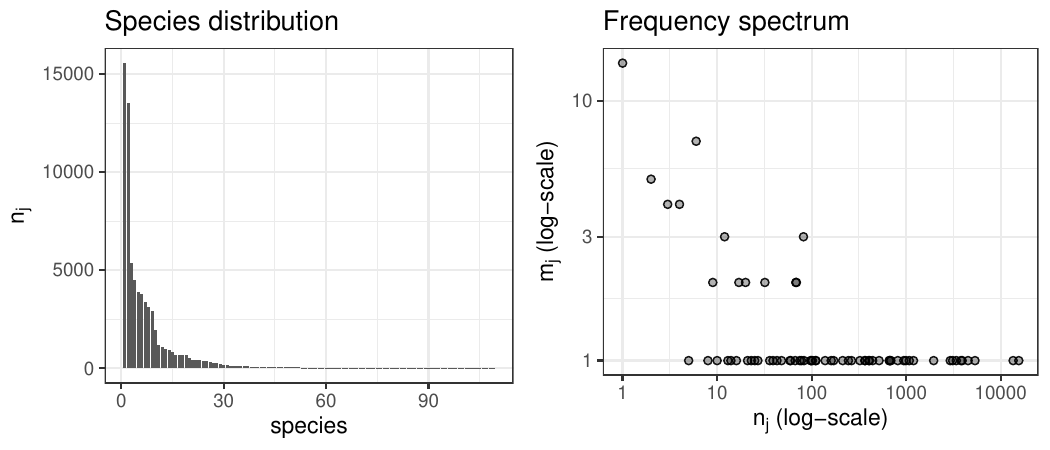}
	\caption{Left panel: distribution of species frequencies, in a decreasing order. Right panel: frequency spectrum.}\label{fig:data}
\end{figure}

Figure~\ref{fig:data} shows the distribution of the species frequencies, in a decreasing order, along with the frequency spectrum. The latter shows on the x-axis the species frequencies $n_j$ and on the y-axis the number of observed species $m_j$ whose frequencies are equal to $n_j$. From the left panel, we can see that the observed species frequencies are far apart from equiprobability of the species, with few of them having large frequencies and an heavy right-tail in their distribution. The right panel shows that most of the species frequencies are observed only once, with few of them observed multiple times, but mainly for small frequencies. Hence, the observed distribution is in support of the assumptions in point (i) of Proposition~\ref{theo:asym_post}.

\subsection{Prior choices for the richness}\label{sec:prior_rich}

A crucial choice regards prior assumptions on the species richness for models \textsc{SDM} and \textsc{DDM}. About the prior specification for the \textsc{SDM} model, apparently it is not possible to find a distribution that produces general explicit expression for the prior Gini--Simpson index moments, since the concentration parameter $\theta$ and the random richness $K$ have a dependence that prohibits any expression simplification.  In the following, we consider as richness prior distribution with the \textsc{SDM} model
\begin{equation}\label{eq:prior_rich_sdm}
	\pi_{\textsc{SDM}}(k) = \frac{\Gamma(2 - \psi_1)\Gamma(2 - \psi_2)}{\Gamma(\psi_1)\Gamma(\psi_2)\Gamma(2 - \psi_1 - \psi_2)} \frac{\Gamma(k + \psi_1 - 1) \Gamma(k + \psi_2 - 1)}{k! (k - 1)!}, \qquad k \geq 1, 
\end{equation}
where $\psi_1, \, \psi_2 > 0$ and $\psi_1 + \psi_2 < 2$. From the previous distribution, $K-1$ describes a beta-negative-binomial distribution with a constraint on the parameter values. A similar but more general class of priors for the number of species was introduced by \citet{Gne10}. For the particular case with $\theta = 1$, it is possible to derive explicitly the first two moments of the Gini--Simpson index assuming the \textsc{SDM} prior models with~\eqref{eq:prior_rich_sdm}. Further details on  $\pi_{SDM}(k)$ and its moments are provided in the appendix.

Unlike model \textsc{SDM}, under \textsc{DDM} prior assumption the influence of the richness $K$ and the parameter $\theta$ on the first two moments of the Gini--Simpson index can be separated. Quite conveniently, from Equations~\eqref{eq:exp_prior_k}-\eqref{eq:var_prior_k}, these moments depend on $K$ only though the first two moments of $1 / K$. For most of the discrete prior distributions, $\mathbb E\large[1 / K\large]$ and $\mathbb E\large[1 / K^2\large]$ do not have a closed form. In the following we consider as prior distribution for the richness a suitable transformation and reparametrization of a beta-geometric probability mass function, with 
\begin{equation}\label{eq:prior_ksdm}
	\pi_{\textsc{DDM}}(k) = \frac{(\gamma_1 + 2) (\gamma_1 + 1)}{\left[ \gamma_1 ( 1 + \gamma_2) + 1 \right] (1 + \gamma_2)} \frac{(\gamma_1\gamma_2 / 2)^{(k -1)} k^2}{\left[2 + \gamma_1(1 + \gamma_2)\right]_{(k)}}, \qquad k \geq 1, 
\end{equation}
with $\gamma_1,\,\gamma_2 > 0$. For this specific prior assumption, we can derive explicit expression of some quantities of interest, as
\[
	\mathbb E\left[ 1 / K \right] = \frac{1}{1 + \gamma_2}, \qquad \overline{\textsc{Var}}\left( 1 / K \right) = 2\frac{(1 + \gamma_1)}{\gamma_1 ( 2 + \gamma_2) + 2}.
\]
Thus, under the prior in equation~\eqref{eq:prior_ksdm}, the expectation and the normalized variance of the Gini--Simpson index became
\[
	\begin{split}
		\mathbb E\left[ S_{\textsc{DDM}} \right] &= \frac{\theta}{1 + \theta} \frac{\gamma_2}{1 + \gamma_2}, \\
		\overline{\textsc{Var}}\left( S_{\textsc{DDM}}\right) &= 1 - \left[ 1 - \frac{2}{(\theta + 2)(\theta + 3)} \right] \left(\frac{1 + \gamma_2}{\theta + \gamma_2 + 1} \right) \left(1 + \frac{\theta}{2 + \gamma_2 + 2 / \gamma_1}\right). 
	\end{split}
\]
The parameter $\theta$ is known to control the index evenness, whereas $\gamma_2$ yields the prior mean of $1 / K$ and therefore controls the centering of the richness prior distribution. Once $\gamma_2$ is fixed, $\gamma_1$ can be chosen specifically to control the dispersion of the Gini--Simpson index, hence providing a flexible model choice, since the variance under the \textsc{DDM} prior can be arbitrary large and its minimum can be shown to be smaller than the one obtained with the \textsc{PD} prior. Other possible prior choices are discussed in the appendix. 

\subsection{Posterior estimates and asymptotics credible intervals}

Starting from the $n = 71\,856$  observations, we consider increasing sample sizes without replacement, to study the behavior of the Gini--Simpson index estimates as far the empirical information grows. More specifically, the samples we consider are of the size $\tilde n_\ell = \big\lfloor n (3 / 10)^{8 - \ell} \big\rfloor$, with $\lfloor \cdot \rfloor$ denoting the floor function, for $\ell = 1, \dots 8$, hence obtaining $n_\ell \in \{31, 93, 281, 852, 2582, 7825, 23712, 71856 \}$. A priori, we consider two classes of models, namely prior distributions with a random number of species and the ones with an infinite number. In particular, for all the model we consider two distinct values for the concentration parameters, with $\theta = 0.1$ or $\theta = 10$. For the \textsc{SDM} and the \textsc{DDM} models, the prior distributions for the richness are the the same introduced in Section~\ref{sec:prior_rich}, with $\psi_1 =  \psi_2 = 0.999$, $\gamma_1 = 108$ and $\gamma_2 = 1$. Both specifications are chosen to set the corresponding expectation close to the empirical value of $1/ k_n$ while keeping the normalized variance large enough. For the \textsc{PD} model, we set the strength parameter $\sigma = 0.25$, as we want to move apart from the \textsc{DP} case, while not falling into degeneracy issue as $\sigma$ approaches $1$. We check through sensitivity analysis for small perturbations around $\sigma = 0.25$ and results were consistent with the one presented hereafter. 

\begin{figure}[h]
	\includegraphics[width = 0.95\textwidth]{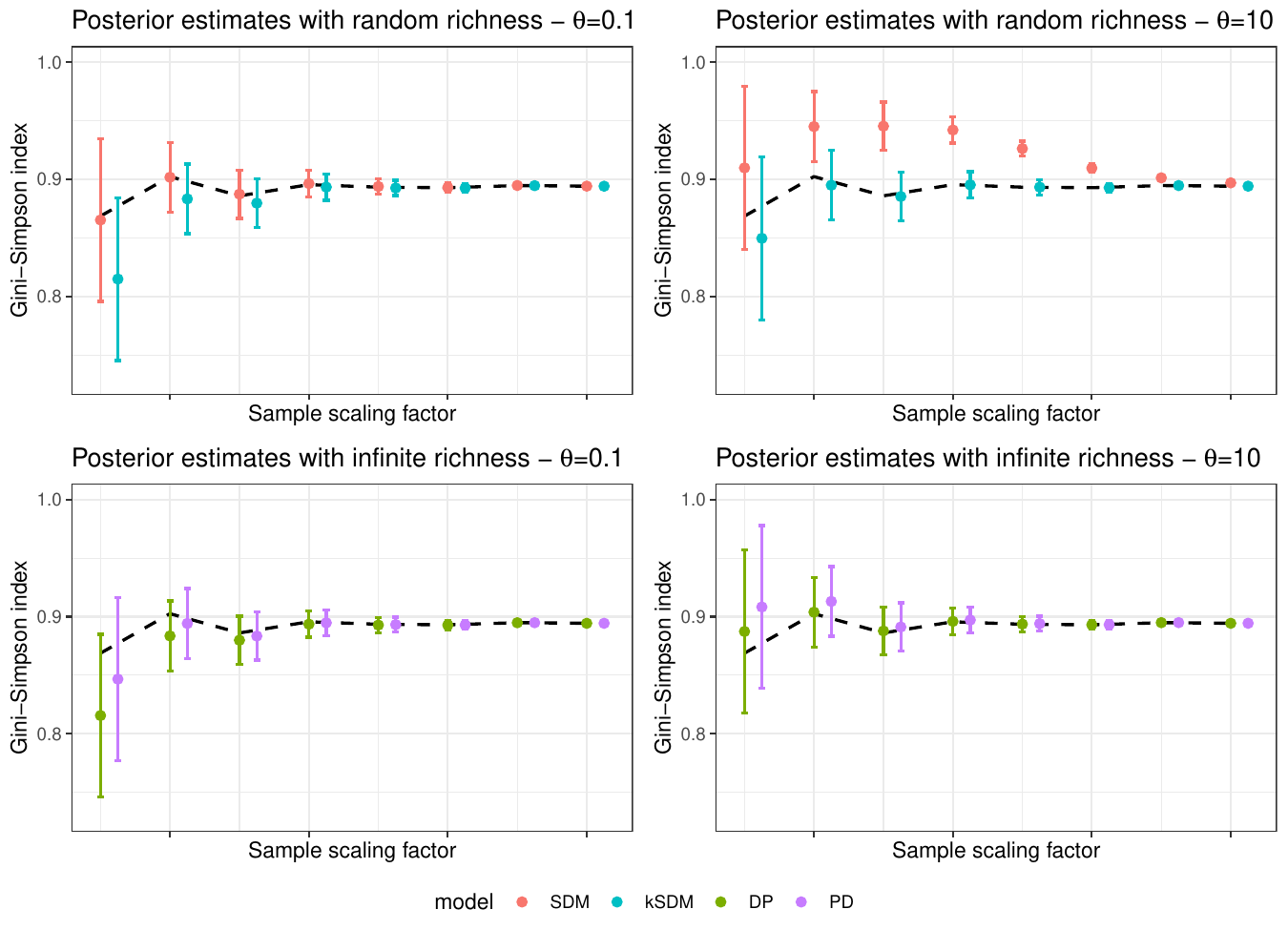}
	\caption{Top line: models with a random number of species. Bottom line: models with an infinite number of species. Left column: small concentration parameter, $\theta = 0.1$. Right column: large concentration parameter, $\theta = 10$. The dashed black line denotes the frequentist estimator $\hat S_n$ of equation~\eqref{eq:freq_est}. }\label{fig:joint}
\end{figure}

Figure~\ref{fig:joint} shows the posterior summaries of the Gini--Simpson index, with different prior settings. Specifically, different colors in the plot corresponds to different model choices, as highlighted in the legend. The top row corresponds to model choices with random richness, while the bottom shows posterior summaries under prior assumptions with an infinite number of species. Points in the plot shows the posterior point estimates, with different models, as far as the sample size increases, while vertical bars denote asymptotic credible intervals. All the models share communalities, as in general large values of the concentration parameter $\theta$ lead to higher prior expectation of the Gini--Simpson index. Further, all the models tends to a common value of the posterior point estimate as far the sample size increases, 
since they all share the same asymptotic behavior.  
Regarding the \textsc{SDM} and \textsc{DDM} models, we remark how the first model can lead to biased estimates, especially for large values of the concentration parameter. More specifically, in the top-right panel, the posterior point estimate of the index start from reasonable values with a small sample size, explodes to large values for moderate sample sizes, and then returns on part of the support close to the frequentist estimator. For the \textsc{DP} and \textsc{PD} model, the discount parameter increases the expectation of the Gini--Simpson index, especially for small values of $\theta$. However, its effect became negligible as far the sample size increases. 

\section{Discussion}

We presented a detailed analysis of the most common prior models used for species diversity estimates, focusing on the Gini--Simpson index, while introducing a novel specification which shows appealing tractability and simple expressions for its prior and posterior summaries. We also provided a detailed comparison of the relative dispersion of the index with different prior specification, showing the benefits of prior specification with random number of species, as the minimum of the variance is smaller, especially for small values of the index expectation. Further, we provided a detailed characterization of the Gini--Simpson index asymptotic distribution, showing that in the limit behavior choosing a prior among the ones presented in the manuscript is irrelevant, since the index always converges to the same distribution.

Observational studies and empirical datasets are frequently affected by the presence of outliers, which may substantially influence posterior inference on species richness and, consequently, on the Gini–Simpson index. The present framework could be further developed by incorporating improper species sampling models or contaminated Bayesian nonparametric priors, in line with the approach proposed by \citet{Cam24}. Moreover, species are often recorded at specific geographical locations, giving rise to point-referenced data. A natural extension of our investigation on the Gini--Simpson index presented here would involve integrating our results with spatial Bayesian nonparametric priors that allow location-dependent frequencies, as early done by \citet{Rei07} and \citet{Dun08}. Finally, given the favorable properties exhibited by prior distributions with random richness, one can extend the previous results on the Gini--Simpson index within the finite-dimensional setting beyond the symmetric Dirichlet distribution, such as the Pitman–Yor multinomial process \citep{Lij20} or, more in general, finite-dimensional normalized random measures with independent increments \citep{Lij24}. 

%
%

\bibliography{bibliography}

\newpage
\appendix

\section{Preliminary results}
We report in this first section some results regarding different prior specifications for the frequencies distribution, which help us to introduce some notation and are preparatory for the following sections.  

\subsection{Mean and variance of the Gini--Simpson index with a generic Dirichlet prior} 

We preserve the same notation of the article. Let $S_{\textsc{SD}}$ and  $S_{\textsc{DD}}$ denote the Gini-Simpson index equipped with the prior generated by the symmetric Dirichlet and the dynamic Dirichlet, respectively. More generally, let $S_{\textsc{D}}$ denote the the Gini-Simpson index equipped with a generic symmetric Dirichlet prior with parameter $a_k$, possibly dependent of the species richness. Hence, we have
\begin{align} 
	\mathbb E[S_{\textsc{D}}] &= 1-\frac{a_k+1}{k a_k+1}=\frac{(k-1)a_k}{k a_k+1}, \label{eq: mean12} \\
	\var(S_{\textsc{D}})& \frac{2a_k(a_k+1)(k-1)}{(ka_k+3)(ka_k+2)(ka_k+1)^2}, \label{eq:var_sd}\\
	\rvar(S_{\textsc{D}})&=\frac{\var(S_{\textsc{D}} )}{\mean(S_{\textsc{D}} )[ 1- \mean(S_{\textsc{D}} )]} = \frac{2}{(k a_k+2)(k a_k+3)}. \label{eq: nvar_mod1}
\end{align}

\begin{proof}[Proof of the identities \eqref{eq: mean12}--\eqref{eq: nvar_mod1}]

Let $k$ be fixed and let $(P_{1,k},\dotsc,P_{k,k})\sim\textsc{Dir}(a_k, \dots,a_k)$. Hence, $P_{1,k} \sim \textsc{Beta}(a_k,(k-1)a_k)$ and $(P_{1,k},P_{2,k}, 1 - P_{1,k} - P_{2,k})\sim \textsc{Dir}(a_k,a_k,(k-2)a_k)$. It is straightforward to compute the moments
\[
 	\begin{split} 
	\mean\left[P_{1,k}^m \right] &= \frac{\Ga(k a_k)}{\Ga(a_k)\Ga((k-1)a_k)}\int_{(0,1)} x^{a_k-1}(1-x)^{a_k(k-1)-1}x^m \diff x\\ 
	&=\frac{(a_k+m-1)\cdots a_k}{(ka_k+m-1)\cdots ka_k}  \\
	\mean\left[P_{1,k}^mP_{2,k}^n\right]&=\frac{\Ga(ka_k)}{\Ga(a_k)^2\Ga((k-2)a_k)}\int_{(0,1)} x^{a_k-1}y^{a_k-1}(1-x-y)^{a_k(k-2)-1}x^m \diff x\\
	&= \frac{(a_k+m-1) \cdots a_k\, (a_k+n-1)\cdots a_k}{(ka_k+m+n-1)\cdots ka_k}, 
	\end{split}
\]
for every pair $(m,n)$ of  positive integers, and then to compute:
\begin{align} \label{eq: mom} 
	\mean\left[\sum_{j=1}^k P_{j,k}^2\right] &= k \mean\left[P_{1,k}^2\right]=k\frac{(a_k+1)a_k}{(ka_k+1)ka_k}=\frac{a_k+1}{ka_k+1}\\
	\begin{split} \label{eq: mom2} 
	\mean\left[\sum_{j=1}^k P_{j,k}^2\right]^2&=k\mean\left[P_{1,k}^4\right]+k(k-1)\mean\left[P_{1,k}^2P_{2,k}^2\right]\\ 
	&=\frac{(a_k+3)(a_k+2)(a_k+1)}{(ka_k+3)(ka_k+2)(ka_k+1)}+
	\frac{(a_k+1)^2a_k(k-1)}{(ka_k+3)(ka_k+2)(ka_k+1)}
	\end{split}
\end{align}
Subtracting the square of \eqref{eq: mom} from \eqref{eq: mom2}, one obtains:
\begin{equation}
	\begin{split}\label{eq: var_k}
	\var\left(1-\sum_{j=1}^k P_{j,k}^2\right)&=\frac{(a_k+3)(a_k+2)(a_k+1)}{(ka_k+3)(ka_k+2)(ka_k+1)}+
	\frac{(a_k+1)^2a_k(k-1)}{(ka_k+3)(ka_k+2)(ka_k+1)}-\frac{(a_k+1)^2}{(ka_k+1)^2}\\
	&=\frac{2a_k(a_k+1)(k-1)}{(ka_k+3)(ka_k+2)(ka_k+1)^2}.
	\end{split} 
\end{equation}
The identity \eqref{eq: mom} yields the prior conditional mean \eqref{eq: mean12}, whereas Equations \eqref{eq: mean12}--\eqref{eq: nvar_mod1} are obtained combining \eqref{eq: var_k} with \eqref{eq: mean12}. 

\end{proof}

\subsection{Computation of the prior variance of models one and two}

Let us recall some formulas from the paper about models SD, DD, SDM and DDM. Regarding the expectations, we can specialize Equation~\eqref{eq: mean12} to the previous models, obtaining 
\begin{align} 
\mean[S_{\textsc{SD}}] &= 1-\frac{\theta+1}{k\theta+1}=\frac{(k-1)\theta}{k\theta+1}, \label{eq: mean} \\
\mean[S_{\textsc{DD}}]&=\frac{\theta}{1+\theta}(1-1/K), \label{eq: condmean} \\
\mean[S_{\textsc{SDM}}]&= 1-\mean\left(\smallfrac{\theta+1}{K \theta +1}\right) =\mean\left[\smallfrac{(K-1)\theta}{K \theta+1}\right], \label{eq: index_mean1} \\
\mean[S_{\textsc{DDM}}]&= \frac{\atwo}{1+\atwo}[1-\mean(1/K)].\label{eq: index_mean2}
\end{align}
By Equations\eqref{eq:var_sd}-\eqref{eq: nvar_mod1}, the variance and the normalized variance for the SD model equal
\begin{equation}
\begin{split}
\var(S_{\textsc{SD}})&=\frac{2\theta(\theta+1)(k-1)}{(k \theta+3)(k \theta+2)(k \theta+1)^2}, \label{eq: var} \\
\rvar(S_{\textsc{SD}})&=\frac{2}{(k \theta+2)(k \theta+3)}. 
\end{split}
 \end{equation}
Considering that $\var(S_{\textsc{SDM}})= \mean[\var(S_{\textsc{SDM}}\mid K = k)]+\var(\mean[S_{\textsc{SDM}}\mid K = k])$ 
together with \eqref{eq: mean12} and \eqref{eq: var}, we obtain the prior variance of the index for the SD model, that is
\begin{equation}
\var(S_{\textsc{SDM}})= \mean\left[\frac{2 \theta (\theta+1)(k-1)}{(k\theta+3)(k\theta+2)(k\theta+1)^2}\right]
+ \var\left(\frac{(k-1)\theta}{k\theta+1}\right),
\label{eq: var_a_fisso}
\end{equation}
and the normalized variance for the \textsc{SDM} is not available in a simple form. Let now consider the DD model. We denote by $E_k = \mean[S_{\textsc{DD}}]$. 
Hence, by \eqref{eq: mean12}, 
\begin{equation}
E_k=\frac{\theta}{1+\theta}\left(1-\frac{1}{K}\right), \qquad 1-E_k=\frac{\theta/k+1}{1+\theta}.
\label{eq: condmean2}
\end{equation}
Combining the Equations in \eqref{eq: condmean2} with \eqref{eq:var_sd}--\eqref{eq: nvar_mod1},  the variance and the normalized variance of the Gini--Simpson for the DD model equal
\begin{align}
\var(S_{\textsc{DD}})&=\frac{2\theta(1+\theta/K)(1-1/K)}{(\theta+3)(\theta+2)(\theta+1)^2}=\frac{2E_k(1-E_k)}{(\theta+2)(\theta+3)}. \label{eq: var_mod2}\\
\rvar(S_{\textsc{DD}})&=\frac{2}{(\theta+2)(\theta+3)}. \label{eq: condnvar} 
\end{align}
Moreover, one obtains for the DDM case
\begin{equation*}
\var(S_{\textsc{DDM}}) = \mean[\var(S_{\textsc{DDM}}\mid K = k)]+\left(\smallfrac{\atwo}{\atwo+1}\right)^2\var\left(\frac{1}{K}\right),
\end{equation*} 
and then apply the first identity in Equation \eqref{eq: var_mod2} with some algebraic computations, the prior variance of the index for the \textsc{DDM} model equals
\begin{equation}\begin{split}
\var(S_{\textsc{DDM}}) = &\smallfrac{2\theta}{(\theta+3)(\theta+2)(\theta+1)^2}\left(1+(\theta-1)\mean\left[\frac{1}{K}\right]-\theta\mean\left[\frac{1}{K^2}\right]\right)\\ 
&+ \left(\smallfrac{\theta}{\theta+1} \right)^2 \left(\mean\left[\frac{1}{K^2}\right]-\mean\left[\frac{1}{K}\right]^2\right)
\label{eq: var-mod2}
\end{split}\end{equation}
By Equations \eqref{eq: var-mod2}--\eqref{eq: mean12}, we can obtain the following expression \eqref{eq: nvar_mod2} for the normalized variance of the index with the DDM
\begin{equation}\label{eq: nvar_mod2}
\rvar(S_{\textsc{DDM}}) = 1 - \left[1-\smallfrac{2}{(\theta+2)(\theta+3)}\right] \left[1-\frac{\theta\var(1/K)}{(1-\mean[1/K])(1+\theta\mean[1/K])}\right].
\end{equation}

\subsection{Computation of prior mean and variance of the index with the Poisson Dirichlet prior}

We aim to show that
\begin{align}\label{eq: meanPD}
 \mean(\Spd)&=\frac{\te + \sigma}{1+\te},\\
 \label{eq: varnPD}
\rvar(\Spd)&=\frac{2}{(\te+2)(\te+3)}.
\end{align}
With the PD prior, it is particularly convenient to consider the size biased permutation $\{\tilde P_j\}_{j\geq 1}$ of $\{P_j\}_{j\geq 1}$, 
i.e., the sequence of the probabilities of the species listed in order of appearance \citep[see][]{Pit06}. Formally, 
$\tilde P_j=P_{I_j}$ where $\{I_j\}_{j\geq 1}$ is a random permutation of the positive integers $\mathbb Z_+$ such that 
\[
\begin{split}
\prob(I_1=l\mid P_1,P_2,\dotsc)&=P_l, \\
\prob(I_j=l\mid I_1,\dotsc,I_{j-1}, P_1,P_2,\dotsc) &=\frac{P_l}{1-\sum_{i=1}^{j-1}P_{I_i}} \ind_{\{I_1,\dotsc,I_{j-1}\}^\complement}(l)
\end{split}
\]
It is known that the sequence $\{\tilde P_j\}_{j\geq 1}$ admits the well known stick--breaking representation recalled in the paper.
Just considering in particular $\tilde P_1$ and $\tilde P_2$, we can verify \citep{Pit97, Ong05, Pit06} that for an arbitrary integrable function $f$ on $[0,1]$ and $g$ on $[0,1]^2$
\begin{align} \label{eq: size-biased1}
\mean\left[\sum_{i=1}^\infty f(P_i)\right]&= \mean\left[\smallfrac{f(\tilde P_1)}{\tilde P_1}\right]\\ \label{eq: size-biased2}
\mean\left[\sum_{i\neq j}^\infty g(P_i,P_j)\right]&= 
\mean\left[\smallfrac{g(\tilde P_1, \tilde P_2)}{\tilde P_1\tilde P_2}(1-\tilde P_1)\right]
\end{align}
Applying \eqref{eq: size-biased1} and \eqref{eq: size-biased2}, one obtains that
\begin{equation}\label{eq: mom1PD}
 \mean\left[\sum_{j=1}^\infty P_j^2\right]=\mean\Big[\tilde P_1\Big]=\mean\big[W_1\big] =\dfrac{1-\sigma}{1+\te},
\end{equation} 
and
\begin{equation}
	\begin{split}\label{eq: mom2PD}
	\mean\left[\sum_{j=1}^\infty P_j^2\right]^2&=\mean\left[\sum_{j=1}^\infty P_j^4\right]+\mean\left[\sum_{i\neq j}^\infty P_i^2 P_j^2\right]\\
 	&=\mean\left[\tilde P_1^3\right]+\mean\left[\tilde P_1\tilde P_2(1-\tilde P_1)\right]\\
 	&=\mean\left[W_1^3\right]+\mean\left[W_1W_2(1-W_1)^2\right]\\
 	&= \frac{(3-\sigma)(2-\sigma)(1-\sigma)}{(\te +3)(\te+2)(\te+1)}+\frac{(\te+\sigma)(1-\sigma)^2}{(\te +3)(\te+2)(\te+1)}. 
 	\end{split}
\end{equation}
Let $S_{\textsc{PD}}$ be the Gini--Simpson index of the weights of our two parameters Poisson Dirichlet process, i.e.,
\[
	S_{\textsc{PD}}=1-\sum_{j=1}^\infty P_j^2.
\] 
By \eqref{eq: mom1PD}, one obtains Equation \eqref{eq: meanPD}. Further, by \eqref{eq: mom1PD}--\eqref{eq: mom2PD}, with some algebraic computations, one obtains that
\begin{equation}\label{eq: varPD}
 \begin{split}
  \var(S_{\textsc{PD}})&=
  \frac{(3-\sigma)(2-\sigma)(1-\sigma)+(\te+\sigma)(1-\sigma)^2}{(\te +3)(\te+2)(\te+1)} - \frac{(1-\sigma)^2}{(1+\te)^2}\\
  &=\frac{2(1-\sigma)(\te+\sigma)}{(\te+1)^2(\te+2)(\te+3)}.
  \end{split}
\end{equation}
Combination of \eqref{eq: meanPD} with \eqref{eq: varPD} yields Equation \eqref{eq: varnPD}.

\section{Proofs of main results}

Here we report the proofs of the theoretical results reported in the article, along with some additional lemmas. 

\subsection{Proof of Proposition 3.1}

\begin{itemize}
\item[(i)] We first consider the SD model. By Equation \eqref{eq: nvar_mod1} and Chebyshev's inequality, for every $\vep>0$, 
\[
	P(1-S_{\textsc{SD}} >\vep)\leq \frac{\var(S_{\textsc{SD}})}{\vep^2}=\frac{2\theta(\theta+1)(k-1)}{\vep^2(k\theta+3)(k\theta+2)(k\theta+1)^2}. 
\]
Being 
\[
	\sum_{k=1}^\infty \frac{(k-1)}{{(k\theta+3)(k\theta+2)(k\theta+1)^2}}<\infty,
\]
one can apply the first Borel Cantelli's lemma obtaining that $S_{\textsc{SD}}$ converges to one almost surely as $k$ diverges to infinite. Being $\var(S_{\textsc{SD}})$ convergent to zero, $L^2$ convergence is given as well. 

\item[(ii)] Considering the DD model, we aim at proving that $S_2$ converges in distribution to $1-\tsum_{j=1}^\infty Q_j^2$, where $Q_1,Q_2,\dotsc$ are the weights of a Ferguson--Dirichlet Process with parameter $a$ arranged in decreasing order. To this aim, let $\{Q_{1,k},\dotsc, Q_{k,k}\}$ be the decreasing rearrangement of \[\left\{P_{1,k},\dotsc, P_{k-1,k},1-\tsum_{j=1}^{k-1} P_{j,k}\right\}\] and $Q_{j,k}=0$ for $j>k$. 

It is known that $\{Q_{1,k},\dotsc,Q_{m,k}\}$ converge in distribution to $\{Q_{1},\dotsc,Q_m\}$ for any $m\geq 1$, as $k$ diverges to infinite \citep[see][Theorem 2.3 and Exercise 2.2.7]{Pit06}. Both $\{Q_{j,k}\}_{j\geq 1}$ and $\{Q_{j}\}_{j\geq 1}$ are random elements valued into the infinite-dimensional simplex 
\[
	\Delta^\infty =\left\{x_1, x_2, \dots \Big| x_j \geq 0, j\geq 1;\, \sum_{j= 1}^\infty x_i\leq 1 \right\}
\]
equipped with the product topology, i.e., the smallest topology that makes projections continuous.
 
In this framework, the just mentioned convergence of finite dimensional distribution is equivalent to convergence in distribution of $\{Q_{j,k}\}_{j\geq 1}$ to $\{Q_{j}\}_{j\geq 1}$ as $k$ diverges to infinite. Moreover, it is known that the $L^2$ topology is equivalent to the product topology in the Hilbert cube 
\[
	H=\bigotimes_{n=1}^\infty [0,1/n]
\] and therefore the same is true for $\Delta \subset H$. 

Letting  $\phi(\bm x)=1-\tsum_{j= 1}^\infty x_j^2$ for every $\bm x=\{x_j\}_{j\geq 1}$ in $\Delta^\infty$, it is clear that $\phi$ is a continuous map  with respect to the $L^2$ topology being $\phi(\cdot)=1-\norm{\cdot}_{L^2}^2$ where $\norm{\cdot}_{L^2}$ denotes the $L^2$ norm.  Therefore, as a consequence of the convergence in distribution of $\{Q_{j,k}\}_{j\geq 1}$ to $\{Q_{j}\}_{j\geq 1}$, $1-\tsum_{j= 1}^\infty Q_{j,k}^2$ converges in distribution to $1-\tsum_{j= 1}^\infty Q_{j}^2$ as desired.

\item[(iii)] We now prove the third and last statement. By equation \eqref{eq: mean12}, one has that 
\[
	\mean[S_{\textsc{D}}]= \frac{1-1 / k}{1+1/(k a_k)},
\] 
where $S_{\textsc{D}}$ denotes the Gini--Simpson index of a generic symmetric Dirichlet distribution with concentration parameter $a_k$, which as $k$ diverges has limit that is different from zero or one only if $k a_k$ converges to a constant different from zero or one. This is the only case in which $S_{\textsc{D}}$ can converge in distribution to a random variable that is not degenerate at zero or one.   
\end{itemize}


\subsection{Proof of Proposition 3.2}

\begin{itemize}
\item[i)] Let the prior mean $\mean[S_\textsc{SD}]$ of the index, under the \textsc{SD} prior assumption, being fixed to a value $\mu$ in the unit interval. We can write explicitly the concentration parameter $\theta$ as function of the mean value from equation~\eqref{eq: mean}, obtaining 
\[
	\theta = \frac{\mu}{k - 1 - k \mu}, \qquad \mu \in (0, 1). 
\]
for which  $k \neq \frac{1}{1 - \mu}$ since $\theta > 0$ and $\mu \in (0,1)$. Hence, we can write explicitly the variance as 
\[
	\rvar(S_\textsc{SD}) = 2 \left[ \left( \frac{\mu}{(1 - \mu) - 1/k} - 2 \right)\left(\frac{\mu}{(1 - \mu) - 1/k} - 2\right) \right]^{-1}, 
\]	
which is an increasing function of $k$. Therefore, fixing $\mu$, the variance approaches its minimum at the value of $k$ for which $(1 - \mu) - 1 / k$ is minimum, which correspond to $k=\lceil 1/\{1-\mu\}\rceil$, where $\lceil \cdot \rceil$ denotes the ceiling function. Therefore, from \eqref{eq: var} we can write explicitly the functional describing the minimum of the variance as $\mu$ varies over its support, which equals
\begin{equation*}\label{eq: lb_mod1}
	m_{\textsc{SD}}(\mu) = \frac{2\{(1-\mu)g(\mu)-1\}^2}{\{(2-\mu)g(\mu)-2\}\{(3-2\mu)g(\mu)-3\}},
\end{equation*}
where $g(\mu)= \lceil 1/(1-\mu) \rceil$ and $\lceil \cdot \rceil$ denotes again the ceiling function.  
\item[ii)] Similarly to before, let the prior mean $\mean[S_\textsc{DD}]$ of the index being fixed to a value $\mu$ in the unit interval, 
then the normalized prior variance of the index is given by \eqref{eq: condnvar} where
\begin{equation}\label{eq: a_K-deg}
	\atwo = \mu/\{1-\mu-1/k\}, \qquad \mu \in (0, 1), 
\end{equation}
which is a decreasing function of $k$ and therefore $\rvar(S_\textsc{DD})$ is an increasing function of it. Therefore, having fixed the prior mean of the index, the normalized variance reaches its minimum at the value of $k$ that minimize the expression in \eqref{eq: a_K-deg}, or equivalently for the value of $k$ that maximizes $1 - \mu - 1/ k$, subject to $1 - \mu - 1 / k > 0$. Hence, $ k=\lceil 1/\{1-\mu\}\rceil$, where $\lceil \cdot \rceil$ denotes the ceiling function. Therefore, one can obtain from \eqref{eq: condnvar} and \eqref{eq: a_K-deg} the minimum of the normalized variance corresponds to the one of point (i).
\end{itemize}


\subsection{Proof of Proposition 3.3}

\begin{itemize}
\item[(i)] From equations \eqref{eq: index_mean1} and \eqref{eq: index_mean2}, we can appreciate that as $\theta$ diverges, both the indices $S_{\textsc{SDM}}$ and $S_{\textsc{DDM}}$ approach $\mean[1 - 1/K]$, hence showing that the index value is upper--bounded by $1 - \mean[1 / K]$ for both models. On the counterpart, for $\theta$ approaching 0, both the indices approach the value 0, which serves as lower--bound. Therefore, the support of the indices is $(0, 1 - \mean[1 / K])$. 
\item[(ii)]Let us begin with the SDM model. The normalized variance $\rvar(S_{\textsc{SDM}})$ is undefined  and $\var(S_\textsc{SDM})=0$ 
if $\mean[S_\textsc{SDM}] = 0$ or $1$, i.e., if $\pi$ is degenerate at one or infinite or in the limit case of $\theta$ being zero.  Being $\mean[S_\textsc{SDM}]$ fixed, the supremum of $\rvar(S_\textsc{SDM})$ (ore equivalently of $\var(S_\textsc{SDM})$) is approached if the the prior $\pi$ is close to a distribution of a non--degenerate random variable whose only values are one and infinite. In order to make this clear, consider a prior $\pi_m$ of $K$ such that $\pi(k) > 0$ if $k \in \{1, m\}$ and $\pi(k) = 0$ otherwise,  i.e., 
\[
	\pi_m(k)=(1-p_m)\delta_1(k)+p_m\delta_m(k)\qquad\text{with}\qquad 0<p_m<1.
\] 
Formula \eqref{eq: mean} implies that  
\[
	\mean_{\pi_m}[S_\textsc{SDM}]= p_m \frac{\theta (m-1)}{(m\theta+1)},
\] 
and therefore, one can take $p_m=\mu(m \theta +1)/[\theta(m-1)]$ so that $\mean_{\pi_m}[S_\textsc{SDM}]=\mu$ provided that $m$ is large. 
In this scenario,  \eqref{eq: var_a_fisso} implies that
\begin{equation*}
	\var_{\pi_m}(S_\textsc{SDM})=p_m \frac{2\theta(\theta+1)(m-1)}{(m\theta+3)(m\theta+2)(m\theta+1)^2} + \left(\frac{(m-1)\theta}{m\theta+1}\right)^2p_m(1-p_m).
\end{equation*}
Therefore, being $\lim_{m\to\infty}p_m=\mu$, $\var_{\pi_m}(S_\textsc{SDM})$ converges to $\mu(1-\mu)$ as $m$ diverges to infinity. Hence, the normalized variance converges to 1. 

Let us consider now model SDD. The normalized variance $\rvar(S_\textsc{SDD})$ is undefined and $\var(S_\textsc{SDD})=0$ if $\mean[S_\textsc{SDD}]=0$ or $1$. By \eqref{eq: index_mean2}, this happens if (i)  $\atwo$ goes to zero, (ii) $K$ degenerates at one or (iii) $K$ degenerates at infinity with $\atwo$ being infinite. In other cases, the upper bound for $\var(S_2)$ given by $\mean[S_\textsc{SDD}](1-\mean[S_\textsc{SDD}])$ is positive and it is achieved when $S_\textsc{SDD}$ is almost surely equal to zero or one, i.e. when $\atwo$ goes to infinity and $K$ is either one or infinite with probability one. More precisely, consider that as $\atwo$ diverges to infinite,  the normalized variance \eqref{eq: nvar_mod2} converges to $\rvar(1/K)$, and then proceed in a similar way as in the proof relate to the SDM model. Hence, we can consider the prior for $K$
\[
	\pi_m(k)=(1-p_m)\delta_1(k)+p_m\delta_m(k),
\]
with $0<p_m<1$, taking $p_m=(1-\mu)m/(m-1)$ so that $\mean[S_\textsc{SDD}]=\mu$ and then let $m$ diverge to infinity. 

\end{itemize}

\subsection{Proof of Proposition 3.3}
By \eqref{eq: meanPD}, one can see that 
\[
	\te=\frac{\mean[\Spd]-\sigma}{1-\mean[\Spd]}
\] 
and, from equation (7) of the main manuscript, we can conclude the proof of Proposition 3.3.

%

\subsection{Additional lemmas}
The following lemma, which regards the delta method, is similar to Theorem 3.8 in \cite{Van98}, but we consider a sequence $\{ \bm Y_n\}_{n\geq 1}$ of random vectors  instead of the sequence $\{y_n\}_{n \geq 1}$ of real numbers. We prove it from scratch to make this paper self--contained.  

\begin{lemma}\label{l: 1}
Let \(\phi:\BR^m\to\BR\) be a map defined and continuously differentiable in a neighborhood of $\bm y_0$. Let $\bm Z_n$ and $ \bm Y_n$ be random vectors taking their values in the domain of $\phi$. If $r_n(\bm Z_n-\bm Y_n)\stackrel{d}{\rightarrow} \bm Z$ for $n\to \infty$ and $\bm Y_n$ converges in probability to $\bm y_0$, then $r_n(\phi(\bm Z_n)-\phi(\bm Y_n))\stackrel{d}{\rightarrow} \nabla \phi(\bm y_0)\cdot \bm Z$.
\end{lemma}

\begin{proof}
It is sufficient to prove that the difference between $r_n(\phi(\bm Z_n)-\phi(\bm Y_n))$ and $\nabla\phi(\bm y_0)\cdot(r_n(\bm Z_n-\bm Y_n))$ converges to zero in probability. 
For $0<t<1$, and fixed $\bm y$, define $g_n(t)=\phi(\bm Y_n+t \bm y)$. 
If $\bm Y_n$ and $\bm Y_n+\bm y$ belong to a neighbourhood of $\bm y_0$ on which $\phi$ is continuously differentiable,  
then by the mean-value theorem, $g_n(1)-g_n(0)=g_n'(\xi)$ for some $0\leq \xi\leq 1$, i.e.
\[
R_n(\bm y):= \phi(\bm Y_n+\bm y)-\phi(\bm Y_n)-\bm y\nabla\phi(\bm y_0)= \bm y\nabla\phi(\bm Y_n+\xi \bm y)-\bm y\nabla\phi(\bm y_0).
\] 
By the continuity of $\nabla\phi$, there exists for every $\vep>0$ a $\delta>0$ such that \( \norm{\nabla \phi(\bm y)-\nabla\phi(\bm y_0)}<\vep \), for every 
\(\norm{\bm y-\bm y_0} <\delta\). 
So, if $\norm{\bm y}<\delta/2$ and $\norm{\bm Y_n-\bm y_0}<\delta/2$, then $\norm{R_n(\bm y)}<\vep\norm{\bm y}$. Thus, for any $\eta>0$,
\begin{multline*}
\prob(r_n\norm{R_n(\bm Z_n-\bm Y_n)}>\eta)\leq \prob(\norm{\bm Z_n-\bm Y_n}\geq \delta/2)+\prob(r_n\norm{\bm Z_n-\bm Y_n}\vep>\eta) + \prob(\norm{\bm Y_n-\bm y_0}>\delta/2).
\end{multline*} 
The first and third term converge to zero as $n\to \infty$. 
The second term can be made arbitrarily small by choosing $\vep$ small.

\end{proof}

For a matrix $A$, denote by $A^\intercal$ its transpose and by $H_\phi(y)$ the Hessian matrix of $\phi$ evaluated at $\theta$. The following lemma is useful if $\phi$ is stationary at $y$.

\begin{lemma}\label{l: 2}
If the assumptions of Lemma \ref{l: 1} hold true, $\phi\in C^2(\BR^m)$ and $\nabla \phi(\bm y_0)$ is zero, then 
\[
	r_n^2(\phi(\bm Z_n)-\phi(\bm Y_n))\stackrel{d}{\rightarrow} \bm Z^\intercal H_\phi(\bm y_0)\bm Z/2.
\] 
\end{lemma}

\begin{proof}
This lemma can be proved similarly as Lemma \ref{l: 1}, but letting
\begin{align*}
R_n(\bm y):&= \phi(\bm Y_n+\bm y)-\phi(\bm Y_n)-\bm y^\intercal H_\phi(\bm y_0)\bm  y/2\\
&=\bm y\nabla \phi(\bm Y_n)+\bm y^\intercal H_\phi(\bm Y_n+\xi \bm y)\bm y/2-\bm y^\intercal H_\phi(\bm y_0)\bm y/2,
\end{align*}
which is true by Taylor's theorem, 
 for some $0<\xi<1$. 
One also needs to consider that $\nabla \phi(\bm Y_n)$ converges in probability to $\bm y_0$ by continuity of $\phi$.

\end{proof}

\subsection{Proofs of Proposition 5.1}

\begin{itemize}
\item[i)] By Bayes theorem, for a generic symmetric Dirichlet prior with concentration parameter $\alpha_k$, 
\[
P_{1,k},\dotsc,P_{k,k}\mid K = k, X_1,\dotsc,X_n \sim \textsc{Dir}(a_k+n_1,\dotsc,a_k+n_{k_n},a_k,\dotsc,a_k), \qquad k\geq k_n,
\] 
where the sequence $\{a_k\}_{k\geq 1}$ is constantly equal to $\theta$ for the SD model and $a_k=\theta/k$ for the DD model.  

Let $S_{k,n}$ be a random variable whose distribution coincides with the conditional distribution of $S_{SD}$ given $X_1,\dotsc,X_n$. 
Further, $k_n$ is eventually equal to $k_0$, $\prob_0$--a.s., and therefore 
the posterior probability that $K$ is bigger than or equal to $k_0$ is eventually one, $\prob_0$--a.s. 
Hence, we consider $K = k$ to be fixed at a value bigger than or equal to $k_0$. We have  
\[
	g(V_{1,n}, \dots, V_{k,n}) = 1-S_{k,n} = \sum_{j=1}^{k} \left(\frac{V_{j,n}}{C_{k,n}}\right)^2\qquad \text{with} \qquad C_{k,n} = \sum_{j=1}^{k} V_{j,n},
\] 
where $V_{1,n},\dotsc,V_{k,n}$ are independent random variables, and $V_{j,n} \sim \textsc{Gamma}(1, a_k + n_j)$, with the proviso that $n_j = 0$ if $j > k_0$. 
For convenience, let the ``true'' probabilities of the species $p_1,\dotsc,p_{k_0}$ ordered so that, by the strong law of large numbers, $n_j/n$ converges $\prob_0$-a.s. to $p_j$, for each $j=1,\dotsc,k_0$. 
By the central limit theorem, 
\[
	\frac{V_{j,n}-n_j-a_k}{\sqrt{n_j+a_k}} \stackrel{d}{\longrightarrow} N(0, 1), \qquad j=1,\dotsc,k_0.
\] 
Therefore, 
\[
	\frac{V_{j,n}-n_j}{\sqrt{np_j}}\, =\, \sqrt{\frac{a_k+n_j}{np_j}}\, \frac{V_{j,n}-n_j-a_k}{\sqrt{n_j+a_k}}+\frac{a_k}{\sqrt{np_j}}
\]
also converges to a standard Normal random variable, for each $j=1,\dotsc,k_0$, by the strong law of large numbers and Slustky's theorem.  

Let $W_{j,n}=V_{j,n}/n$, so that 
\[
	\begin{split}
		\sqrt{n}\left(W_{j,n}-\frac{n_j}{n}\right) &\stackrel{d}{\longrightarrow} N(0, p_j), \qquad j = 1, \dots, k_0,\\
		\sqrt{n}\left(W_{j,n}-\frac{\sqrt{a_k}}{n}\right) &\stackrel{d}{\longrightarrow} \delta_0, \qquad j > k_0.
	\end{split}
\] 
Further, we note that $1-S_{k,n}$ is distributed as $g(W_{1,n},\dotsc,W_{k,n})$. Denote by 
\begin{equation}\label{eq:mukn}
	\mu_n(k) = g(n_1/n,\dotsc,n_{k_0}/n,\sqrt{a_k}/n,\dotsc,\sqrt{a_k}/n)= \frac{1-\hat S_n+(k-k_0)a_k/n^2}{1+(k-k_0)\sqrt{a_k}/n},
\end{equation}
where 
\[
	\hat S_n=1-\sum_{j=1}^{k_n}\left(\frac{n_j}{n}\right)^2, \qquad \text{with}\qquad \hat S_n \stackrel{p}{\longrightarrow} 1-\sum_{j=1}^{k_0} p^2_j.
\]

To find the asymptotic distribution of $S_{k,n}$, the delta method can be applied.  
More precisely, Lemma \ref{l: 1} can be applied with $\bm Z_n=(W_{1,n},\dotsc,W_{k,n})$, 
$\phi=g$, $r_n=\sqrt{n}$, $\bm Y_n=(n_1/n,\dotsc,n_{k_0}/n,\sqrt{a_k}/n,\dotsc,\sqrt{a_k}/n)$ and $\bm y_0=(p_1,\dotsc,p_{k_0},0,\dotsc,0)$, 
obtaining that    
\[
	\sqrt{n}\big(1-S_{k,n}-\mu_n(k)\big) \stackrel{d}{\longrightarrow} N(0, \sigma^2)
\]
being  
\[
	\sigma^2=4\sum_{j=1}^{k_0} p_j\left(p_j-\sum_{\ell=1}^{k_0}p_\ell^2\right)^2=4\sum_{j=1}^{k_0} p_j^3-4\left(\sum_{j=1}^{k_0} p_j^2\right)^2.
\]
In fact, \(\sqrt{n}(1-S_{k,n}-\mu_n(k))\) converges to  \(\bm Y^\intercal \nabla g(p_1,\dotsc,p_{k_0},0,\dotsc,0),\)  where $\bm Y$ is a $k$--dimensional vector of independent 0-mean normal random variables, whose $j$--th component $Y_j$ has variance $p_j$, for $j=1,\dotsc,k_0$ and zero variance for $j=k_0+1,\dotsc,K$, and $\nabla g$ denotes the gradient of $g$. Hence, the limiting distribution of \(\sqrt{n}(1-S_{k,n}-\mu_n(k))\) is normal with zero mean 
and its variance is the weighted sum of the squares of the first order partial derivatives of $g$ evaluated at $(p_1,\dotsc,p_{k_0},0,\dotsc,0)$, 
where the weights are $p_1,\dotsc,p_{k_0}$. The sequence \( \sqrt{n}(1-\hat S_n-\mu_n(k)) \) converges in distribution to zero, 
and therefore by Slustky's theorem, 
\[
	\sqrt{n}(S_{k,n}-\hat S_n)= \left[ \sqrt{n}\big(1-\hat S_n-\mu_n(k)\big) -\sqrt{n}\big(1-S_{k,n}-\mu_n(k)\big)\right] \stackrel{d}{\longrightarrow} N(0, \sigma^2). 
\]

To conclude, we need to prove the convergence in distribution of \(\sqrt{n}(S_{n}-\hat S_n)\). To this aim, let \(\{R_n\}_{n\geq 1}\) be a sequence of random variables such that $R_n = \pi_n(k)/\pi(k)$ with probability \(\pi(k)\) for each \(k\) in the support of \(\pi\), where $\pi_n$ denotes the posterior distribution of $K \mid X_1, \dots, X_n$. Moreover, let 
\[
	R = \pi(k_0) \delta_{1 / \pi(k_0)} + [1-\pi(k_0)] \delta_0.
\]

With the SDM model, the posterior $\pi_n$ is consistent \citep{Bis13}. The result provided by \cite{Bis13} cannot be applied for the DDM model, 
since the parameter $a_k$ of the symmetric Dirichlet distribution of $(P_{1,k},\dotsc,P_{k,k})$ depends on $k$. Nevertheless, it is not difficult to verify that $\pi_n$ is consistent. Indeed, one can apply Lemma 1 in \cite{Bis13}, showing that 
\begin{equation}\label{eq: to_prove1}
	\lim_{n\to \infty} \sum_{\ell>k_0} \pi(\ell) \dfrac{\ell!}{(\ell-k_0)!\Ga(\atwo/\ell)^{k_0}}  \prod_{j=1}^{k_0} \frac{\Ga(np_j+\atwo/\ell)}{\Ga(np_j+\atwo/k_0)}=0.
\end{equation}
To this aim, one can apply the dominated convergence theorem. One can verify that the general term of the series in \eqref{eq: to_prove1} is bounded by $\pi(l)$, for large $n$, by recalling that the Gamma function is increasing on $(2,\infty)$, and therefore each factor in the product in \eqref{eq: to_prove1} is less than one if $n$ is large, and noting that 
\[
	\Ga(x)\geq \frac{1}{x}-\frac{1}{x+1} = \frac{1}{x(x+1)}, \qquad \text{for every} \qquad x>0, 
\]
being $e^{-t}\geq (1-t)\ind_{(0,1)}(t)$, for every $t$. Moreover, one can apply the Stirling's formula 
to verify that 
\[
	\lim_{n\to \infty} \prod_{j=1}^{k_0} \frac{\Ga(np_j+\atwo/\ell)}{\Ga(np_j+\atwo/k_0)} = 0
\] 
Hence, $\pi_n$ is consistent for both models SDM and DDM. Therefore, \(R_n\) converges in distribution to \(R\) and the mean value of \(R_n\) comverges to one, which is the mean value of \(R\). Being 
\[ 
	0\leq \prob\left(\sqrt{n}(S_{k,n}-\hat S_n)\leq x\right)\frac{\pi_n(k)}{\pi(k)}\leq R_n,  
\]
one can apply Pratt's lemma to obtain that 
\[
	\prob\left(\sqrt{n}(S_{n}-\hat S_n)\leq x\right) = \sum_{k\geq 1} \prob\left(\sqrt{n}(S_{k,n}-\hat S_n)\leq x\right)\pi_n(k)
\]
converges to the cumulative distribution function of a random variable $N(0, \sigma^2)$.

Lastly, we need to show that the same result holds also for the PD case. The posterior distribution of a Poisson--Dirichlet process with parameter $(\sigma, \theta)$ is the same one of the process given by
\begin{equation}\label{eq: postPDproc}
	\eta_n Q + \sum_{j=1}^{k_n} \eta_{n,j} \delta_{X_j^*}, 
\end{equation}
where $X_1^*,\dotsc, X_{k_n}^*$ are the distinct species observed in $X_1,\dotsc,X_n$ with frequencies
$n_1, \dots, n_{k_n}$. Further, we have that  
$(\eta_n, \eta_{n,1},\dotsc,\eta_{n,k_n})\sim \textsc{Dir}(\theta+k_n\sigma, n_1-\sigma,\dotsc,n_{k_n}-\sigma)$, and $Q$ is a Poisson--Dirichlet process with parameter $(\sigma, \te +k_n\sigma)$ \citep[][Corollary 20]{Pit96}. Moreover, $Q$ and $(\eta_{n,1},\dotsc,\eta_{n,k_n},\eta_n)$ are stochastically independent. Letting $Q_1,Q_2,\dotsc$ be the weights of $Q$, the posterior distribution of $S_{PD}$ is the same distribution of the following random variable
\begin{equation*}\label{eq: posteriorPD} 
	1-\eta_n^2\sum_{l\geq 1} Q_l^2 -  (1-\eta_n)^2\sum_{i=1}^{k_n} \xi_{n,i}^2,  
\end{equation*}
where $\xi_{n,i}=\eta_{n,i}/(1-\eta_n)$, for every $i=1,\dotsc,k_n$. By the properties of the finite Dirichlet distribution,  $\eta_n$ and $(\xi_{n,1},\dotsc,\xi_{n,k_n})$ are stochastically independent, being the distribution of $(\xi_{n,1},\dotsc,\xi_{n,k_n-1})$ finite dimensional Dirichlet with parameter $(n_1-\sigma,\dotsc,n_{k_n}-\sigma)$. Finally, note that the thesis of Propositions 5.1 is also valid  in the case of PD model with $a=-\sigma$ and $\sigma\in(0,1)$.  This implies that the asymptotic distribution of  $1- \tsum_{i=1}^{k_n} \xi_{n,i}^2$ is the same one given before. Since $\eta_n\sim\textsc{Beta}(\theta+k_n\sigma, n-k_n\sigma)$ converges to zero in distribution, by Slustky's theorem, the thesis follows also for the PD model.

\item[ii)] The proof is similar to point i) and the same notation will be used. As before,  let $k$ be fixed at a value bigger than or equal to $k_0$.
First, we study the convergence of \(n(1-S_{k,n}-\mu_n(k))\), where $\mu_n(k)$ is defined in Equation~\eqref{eq:mukn}. This can be done using the delta method based on the second order Taylor expansion, being $g$ stationary at $(1/k_0,\dotsc,1/k_0)$. 
More precisely, Lemma \ref{l: 2} can be applied with $\bm Z_n=(W_{1,n},\dotsc,W_{k,n})$, 
$\phi=g$, $r_n=\sqrt{n}$, 
\[
	\bm Y_n=(n_1/n,\dotsc,n_{k_0}/n,\sqrt{a_k}/n,\dotsc,\sqrt{a_k}/n)
\] and $\bm y_0=(1 / k_0,\dotsc,1 / k_0,0,\dotsc,0)$.
So, 
\[
	n(1-S_{k,n}-\mu_n(k))\stackrel{d}{\longrightarrow} \frac{1}{2} \bm Z^\intercal H_g \bm Z,
\]
where $H_g$ denotes the Hessian matrix of $g$ evaluated at $(1/k_0,\dotsc,1/k_0,0,\dotsc,0)$, and $\bm Z=\sqrt{\frac{1}{k_0}} \bm T$, where $\bm T$ is a $k$-dimensional random vector whose first $k_0$ components are independent standard normal random variables and the other ones are identically zero.

Clearly, $\bm Z^\intercal H_g \bm Z$ doesn't change if the last $k-k_0$ elements of the $\bm Z$ vector are deleted and the last $k-k_0$ rows and columns of $H_g$ are deleted. Computations yield the elements of the matrix $H_g$, that equals 
\[
	[H_g]_{i,j}= \begin{cases} -2/k_0, \qquad &\text{if } i \neq j \\ 2(1-1/k_0), \qquad &\text{if } i = j \end{cases}, \qquad i,j=1,\dotsc,k_0. 
\]
Hence, 
\[
	\bm T^\intercal H_g  \bm T=2\sum_{j=1}^{k_0} (T_j- \bar{T})^2 \sim 2 \chi^2_{k_0 - 1}, 
\] 
where $\bar{T}=\tsum_{j=1}^{k_0} T_j/n$. To check the distributional identity, we first note that  $\frac{1}{2} H_g$ is a projection matrix. Indeed, it is the difference of the identity matrix and a matrix with all entries equal to $1/k_0$ and it holds the idempotent property, i.e.,
\[
	\frac{1}{2} H_g = \left(\frac{1}{2} H_g\right)^2.
\] 
Therefore, it can be written in spectral form, as \(\frac{1}{2} H_g=PDP^\intercal\), so that 
\[
	\frac{1}{2} \bm T^\intercal H_g \bm T= \bm T^\intercal P D P^\intercal \bm T.
\] 
Since $P$ is orthogonal, $P^\intercal \bm T$ is 
a vector of independent standard normal random variables like $\bm T$. 
Further, \(\frac{1}{2}  \bm T^\intercal  H_g \bm T\) is distributed as \(\bm T^\intercal D \bm T\), or equivalently as \(\tsum_{i=1}^{k_0} \lambda_{i} T^2_i\), being $T_1,\dotsc, T_{k}$ independent standard normal random variables or identically zeros, and $\lambda_1, \dots, \lambda_k$ 
the eigenvalues of  $\frac{1}{2} H_g$. Being $\frac{1}{2} H_g$ a projection matrix, its eigenvalues, i.e. the $\lambda_i$'s, are ones and zeros. Moreover, their sum is the trace of $\frac{1}{2} H_g$, that is $k_0-1$. As a result, the distribution of \(\bm T^\intercal D \bm T\) is a chi squared with $k_0-1$ degrees of freedom and \(nk_0(1-S_{k,n}-\mu_n(k))\) 
converges in distribution to a chi squared with $k_0-1$ degrees of freedom.  

Letting 
\[
	\nu_n(K)=\sum_{j=1}^{k_n} \left(\frac{n_j}{n}\right)^2 \frac{1}{1+(k-k_0)\sqrt{a_k}/n},
\]
the sequence \( n(\mu_n(k)-\nu_n(k))\) converges in distribution to zero, 
and therefore by Slustky's theorem, 
\( n(1-S_{k,n}-\nu_n(k))= n(1-S_{k,n}-\mu_n(k)) + n(\mu_n(k)-\nu_n(k)) \) converges to a \(\chi^2_{k_0-1}\) multiplied by $1/k_0$.  As we have done for point i), we can apply Pratt's lemma to show that, letting $\pi_n$ be the posterior distribution of $K$, 
\[
\prob\left(nk_0(1-S_{n}-\nu_n(k))\leq x\right) =  \sum_{k\geq 1} \prob\left(nk_0(1-S_{k,n}-\mu_n(k))\leq x\right)\pi_n(k),
\]
converges to the cumulative distribution function of $\chi^2_{k_0-1}$. In other words, \(n(1-S_{n}-\nu_n(k))\) converges in distribution to a \(\chi^2_{k_0-1}\) multiplied by $1/k_0$. Since $K$ converges in distribution to $k_0$, being $\pi_n$ its distribution, 
\[
	n(\nu_n(k)-\nu_n(k_0)) = -(k-k_0)\sqrt{a_K}\sum_{j=1}^{k_n} \left(\frac{n_j}{n}\right)^2 \frac{1}{1+(k-k_0)\sqrt{a_k}/n}
\] 
converges in distribution to zero. Therefore, by Slustky's theorem, 
\[
	nk_0\big(1-S_{n}-\nu_n(k_0)\big)=nk_0\big(1-S_{n}-\nu_n(k)\big)+ nk_0\big(\nu_n(k)-\nu_n(k_0)\big)
\] 
converges to $\chi^2_{k_0-1}$, where \(\nu_n(k_0)=1-\hat S_n\). Finally, we can use the same strategy of the last part of point i), with the posterior representation provided in Equations~\eqref{eq: posteriorPD}  and \eqref{eq: postPDproc}, to show that the same result holds for the PD model. In particular, the asymptotic distribution of  $1- \tsum_{i=1}^{K_n} \xi_{n,i}^2$ is the same one given before for the equiprobable case with a generic symmetric Dirichlet prior distribution. Since $\eta_n\sim\textsc{Beta}(\theta+k_n\sigma, n-k_n\sigma)$ converges to zero in distribution, by Slustky's theorem, the thesis follows also for the PD case. 

\end{itemize}


\subsection{Proof of Corollary 5.1}
Let 
\[
	\sigma^2=4\sum_{j=1}^{k_0} p_j\left(p_j-\sum_{\ell=1}^{k_0}p_\ell^2\right)^2=4\sum_{j=1}^{k_0} p_j^3-4\left(\sum_{j=1}^{k_0} p_j^2\right)^2.
\]
We note that
\[ 
	\sqrt{\frac{n}{s^2_n}}(S_n-\hat S_n) = \sqrt{\frac{\sigma^2}{s^2_n}} \left[\sqrt{\frac{n}{\sigma^2}}(S_n-\hat S_n)\right], 
\]
where the random sequence given by the factor within the square brackets converges in distribution to a standard normal random variable by Proposition 5.1, and the one given by the first factor converges to one almost surely (and therefore in distribution) by consistency of $s^2_n$.
By Slustky's Theorem, 
\[
	\sqrt{\frac{n}{s^2_n}}(S_n-\hat S_n) \stackrel{d}{\longrightarrow} N(0,1).
\]

\section{Computation of exact estimates}

We report in the following subsections the calculation to derive the posterior point estimates for all the models considered in the manuscript. 

\subsection{Finite and random number of species (\textsc{SD, DD, SDM} and \textsc{DDM} models)}

We first note that, for models \textsc{SD} and \textsc{DD} with $k$ species, the posterior distribution of the Gini--Simpson index is obtained by the the conditional distribution of the species frequencies $P_1,\dotsc,P_{k}$ given the sample $X_1, \dots, X_n$, that by Bayes theorem equals
\[
	P_1, \dots, P_k \mid X_1, \dots, X_n \sim \textsc{Dir}(a_k + n_1, \dots, a_k + n_k), 
\]
where $a_k=\theta$ for model \textsc{SD} and $a_k=\theta/k$ for model \textsc{DD}, 
$n_j$ is the number of times we observe the $j$--th species in the sample, with $n_j = 0$ if the $j$--th species is not present in the sample.

In order to obtain the posterior mean of the index, recall that for a random variable $W\sim \textsc{Beta}(\al, \be)$, 
\[
	\mean\Big[W^2\Big]=\frac{\al(\al+1)}{(\al+\be)(\al+\be+1)}.
\] 
Hence, one can obtain the expression of the posterior mean of the index as
\begin{equation}\begin{split}\label{eq: posterior_mean2}
	\mean[S \mid X_1, \dots, X_n] = 1-&\sum_{j=1}^{k} \smallfrac{(a_k+n_j+1)(a_k+n_j)}{(ka_k+n+1)(ka_k+n)}\\
	&= \frac{n(n-1)}{(ka_k+n+1)(ka_k+n)}{\tilde S}_n+\frac{(ka_k+1)(ka_k+2n)}{(ka_k+n+1)(ka_k+n)}\frac{a_k(k-1)}{ka_k+1}\\
	&=p_{k,n}{\hat S}_n+ (1-p_{k,n})\mean[S], 
\end{split}\end{equation}
where 
\[
	\tilde S = 1 - \sum_{j=1}^{k_n} \frac{n_j(n_j - 1)}{n(n-1)} 
\]
denotes the unbiased frequentist estimator with minimum variance for the Gini--Simpson index, with $k_n$ being the number of distinct species observed in the sample. Specifically, we can explicit the weights in Equation~\eqref{eq: posterior_mean2} for the \textsc{SD} and \textsc{DD} models, obtaining
\[
	\begin{split}
		p_{k,n}^\textsc{SD} &= \frac{n(n-1)}{(k \theta + n + 1)(k \theta + n)}, \qquad \text{for model }\textsc{SD}, \\
		p_{k,n}^\textsc{DD} &= \frac{n(n-1)}{(\theta + n + 1)(\theta + n)}, \,\quad\qquad \text{for model }\textsc{DD}. \\
	\end{split}
\]

While mixing the model with respect to the number of species, for the \textsc{SDM} and \textsc{DDM} models, from the conditional expectation of the index we have
\begin{equation*}\label{eq: posterior_mean}
	\begin{split}
	\mean[S\mid X_1, \dots, X_n] = 1-\sum_{k=K_n}^{\infty}\sum_{j=1}^{k} \smallfrac{(a_k+n_j+1)(a_k+n_j)}{(ka_k+n+1)(ka_k+n)} \pi_n(k) ,
	\end{split}
\end{equation*}              
where $\pi_n$ denotes the posterior of $K$,  that is given by
\begin{equation*}
	\pi_n(k)\propto \frac{k!}{(k-k_n)!}\mean\left[\prod_{j=1}^{K_n}P_{j,k}^{n_j}\right]\pi(k) \ind_{[k > k_n]} = \frac{k!}{(k-k_n)!}\frac{\Ga(ka_k)}{\Ga(n+ka_k)}\prod_{j=1}^{k_n}\smallfrac{\Ga(n_j+a_k)}{\Ga(a_k)}\pi(k)\ind_{[k > k_n]} ,
\end{equation*}
In particular, we obtain 
\begin{equation*}\label{eq: postK}
	\begin{split}
		\pi_n(k)&\propto \frac{k!}{(k-k_n)!}\frac{\Ga(k\theta)}{\Ga(n+k\theta)}\pi(k)\ind_{[k > k_n]} , \,\;\;\qquad\qquad\qquad\qquad \text{for model }\textsc{SDM},\\
		\pi_n(k)&\propto \frac{k!}{(k-k_n)!}\frac{1}{\Ga(\atwo/k)^{K_n}} \prod_{j=1}^{K_n}\Ga(n_j+\atwo/k) \pi(k)\ind_{[k > k_n]} , \qquad \text{for model }\textsc{DDM}.
	\end{split}
\end{equation*}
After manipulation, the posterior expectation can be rewritten as 
\[
	\mean[S\mid X_1, \dots, X_n]  = q_n \tilde S_n + (1 - q_n) W_n, 
\]
where for the weights we have
\begin{alignat*}{2}
	q_{n}^\textsc{SDM} &= n(n-1) \mean_{\pi_n}\left[\frac{1}{(K \theta + n + 1)(K \theta + n)}\right], \qquad &&\text{for model }\textsc{SD}, \\
	q_{n}^\textsc{DDM} &= \frac{n(n-1)}{(\theta + n + 1)(\theta + n)},&&\text{for model }\textsc{DD}, \\
\end{alignat*}
and for the prior rough estimator
\begin{alignat*}{2}
	W_{n}^\textsc{SDM} &= 1 - \frac{\theta + 1}{n} \left( \frac{n \tau_n + 1}{1 - q_n^\textsc{SDM} }- 1\right), \qquad && \text{for model }\textsc{SD}, \\
	W_{n}^\textsc{DDM} &= \frac{\theta}{1 + \theta}\left(1 - \mean_{\pi_n}\left[\frac{1}{K}\right] \right), && \text{for model }\textsc{DD}, \\
\end{alignat*}
with $\tau_n = \mean_{\pi_n}\left[(\theta K + n)^{-1}\right]$. 
\subsection{Infinite number of species (\textsc{DP} and \textsc{PD} models)}

The posterior distribution of a Poisson--Dirichlet process with parameter $(\sigma, \theta)$ \citep[][Corollary 20]{Pit96} is the same one of the process given by
\[ 
	\eta_n Q + \sum_{j=1}^{k_n} \eta_{n,j} \delta_{X_j^*}, 
\]
where $X_1^*,\dotsc, X_{k_n}^*$ are the distinct values of the observations $X_1,\dotsc,X_n$,  $n_j$ is the number of observations $X_i$ equal to $X_j^*$, $j=1,\dotsc,k_n$, 
\[
	(\eta_n, \eta_{n,1},\dotsc,\eta_{n,k_n}) \sim \textsc{Dir}(\theta+k_n\sigma, n_1-\sigma,\dotsc,n_{k_n}-\sigma),
\] 
and $Q\sim PD(\sigma, \te +k_n\sigma)$. Moreover, $Q$ and $(\eta_n, \eta_{n,1},\dotsc,\eta_{n,k_n})$ are stochastically independent. 
Letting $Q_1,Q_2,\dotsc$ be the weights of $Q$, 
the posterior distribution of $S$ is the same distribution of the following random variable
\begin{equation}\label{eq: posteriorPD} 
	1-\eta_n^2\sum_{l\geq 1} Q_l^2 -  (1-\eta_n)^2\sum_{i=1}^{k_n} \xi_{n,i}^2,  
\end{equation}
where $\xi_{n,i}=\eta_{n,i}/(1-\eta_n)$, for every $i=1,\dotsc,k_n$. 
By the properties of the finite Dirichlet distribution, $\eta_n$ and $(\xi_{n,1},\dotsc,\xi_{n,k_n})$ are stochastically independent, 
being  
\[
	(\xi_{n,1},\dotsc,\xi_{n,k_n}) \sim \textsc{Dir}(n_1-\sigma,\dotsc,n_{k_n}-\sigma). 
\]
By Equation~\eqref{eq: meanPD}, 
\[
	\mean\left[\sum_{l\geq 1} Q_l^2\right]=\frac{1-\sigma}{1+\te+k_n\sigma}.
\]
Hence, the posterior mean of $S_\textsc{PD}$ is given by
\begin{equation*}
\mean[S_\textsc{PD} \mid X_1, \dots, X_n] = 1-\smallfrac{(\te+k_n\sigma)(1-\sigma)}{(\te+n+1)(\te+n)}-
\sum_{j=1}^{k_n}\smallfrac{(n_j-\sigma+1)(n_j-\sigma)}{(\te+n+1)(\te+n)},
\end{equation*} 
that is equal to a linear combination of the prior mean $\mean[S_\textsc{PD}]$ and the frequentist unbiased estimator $\hat S_n$, 
\begin{equation}\label{eq: postmeanPD}
	\begin{split}
		\mean[S_\textsc{PD} \mid X_1, \dots, X_n] &= \frac{(1+\te)(\te+2n)}{(\te+n+1)(\te+n)}\frac{\te+\sigma}{1+\te}+\frac{n(n-1)}{(\te+n+1)(\te+n)}
\left[ 1-\sum_{j=1}^{k_n}\smallfrac{n_j(n_j-1)}{n(n-1)}\right] \\
&=q_n^\textsc{PD} \hat S_n + (1 - q_n^\textsc{PD}) W_n^\textsc{PD}.
	\end{split}
\end{equation}
The estimator in Equation~\eqref{eq: postmeanPD} is quite similar to the posterior mean with the \textsc{DD} model,
which is also a convex linear combination with the same weights. Hence, the estimator \eqref{eq: postmeanPD} is also a smooth alternative of $\tilde{S}_n$. Nevertheless, if one wishes to compare the two models fixing the prior mean of the index to the same value, then one notices that the weight $q_n^\textsc{DD}$ of the frequentist estimator $\tilde{S}_n$ is necessarily higher with the Poisson Dirichlet prior.  
But the main difference between the two estimators is that in \eqref{eq: postmeanPD} the rough estimator $W_n^\textsc{DD}$ is replaced by the prior mean.


\section{Computation of the curves of the minimum normalized variance for \textsc{SDM} and \textsc{DDM} models}
In the paper (Figure 2), a plot shows the curves of the minimum normalized variance under the different prior models. In this section, it will be shown 
how such minimums have been computed for the \textsc{SDM} and \textsc{DDM} models, being the prior mean of the index a fixed value $\mu\in[0,1]$.

\begin{itemize}
\item[-] \textbf{\textsc{DDM} model -}
Being the prior mean fixed and equal to $\mu$, by Equation~\eqref{eq: index_mean2}, $\mean[1/K]$ is function of $\mu$ and of the parameter $\te$, namely $\mean[1/K]= 1 - \mu(1+\te)/\te$. Hence, by Equation~\eqref{eq: nvar_mod2}, $\rvar(S_2)$ is the following function of $\te$ and $\var(1/K)$
\[
	V_\textsc{DDM}(\te,\var(1/K))= 1-\left(1-\smallfrac{2}{(\te+2)(\te+3)}\right)\left(1-\frac{\var(1/K)}{\mu(1-\mu)}\left(\frac{\te}{\te+1}\right)^2\right).
\] 
For a fixed  value of $\te$, $V_2$ is an increasing function of $\var(1/K)$ and therefore the minimum is achieved when $\var(1/K)$ is minimum. If the reciprocal of $\mean[1/K]$ is integer, then $\var(1/K)$ can be zero with $K$ degenerate at $1/\mean[1/K]$. In general, the minimum of $\var(1/K)$ is achieved if the distribution of $K$ is concentrated on the two integer values which are closest to $1/\mean[1/K]$, i.e.\ $k^\star = \lfloor 1/\mean[1/K] \rfloor$ and $k^\star+1$. In order ensure that $\mu$ is the prior mean of the index, it must be 
\[
	\pi(\{k^\star+1\})=1-\pi(\{k^\star\})=k^\star(k^\star+1)\left(\frac{1}{k^\star}-\mean\left[\frac{1}{K}\right] \right),
\]   
which is zero if $1/\mean[1/K]$ is integer and, therefore, the prior is degenerate at $1/\mean[1/K]$.  For such prior, 
$\var(1/K)=v(\mean[1/K])$, where the function $v$ equals 
\[
	v(x)=\left(x-\frac{1}{k^\star(x) +1}\right)\left(\frac{1}{k^\star(x)}-x\right),
\] 
and $k^\star(x)=\lfloor 1/x \rfloor$. Therefore, the minimum normalized variance for model two has been obtained as
\[
	\min_{x\in[0,1-\mu]} V_\textsc{DDM}\left(\frac{\mu}{1-\mu-x},v(x)\right),
\]
since $0\leq \mean[1/K]\leq 1-\mu$. Such minimum has been computed as $\mu$ varies over the grid of 100 equally spaced points in the unit interval, 
and then plotted in Figure 2 of the manuscript. 

\item[-] \textbf{\textsc{SDM} model -} The procedure with the \textsc{SDM} model is quite more involved and it requires the following preliminary lemma.
\begin{lemma}\label{l: convex}
If $\phi$ is a non negative convex function, and $\{x_1,x_2,\dotsc\}$ is a (finite or infinite) closed set, 
then the minimum of $\mean[\phi(X)]$ over all discrete random variables $X$ with the same mean $\mu$ and valued into $\{x_1,x_2,\dotsc\}$  
is $t\phi(\tilde x_1)+(1-t)\phi(\tilde x_2)$, where 
$\tilde x_1$ and $\tilde x_2$ are the two values among $x_1,x_2,\dotsc$ that are closest to $\mu$, i.e.\ 
$\tilde x_1=\max\{j: x_j\leq \mu\}$, $\tilde x_2=\min\{j: x_j\geq \mu\}$, and $t=(\tilde x_2-\mu)/(\tilde x_2-\tilde x_1)$. 
\end{lemma}
\begin{proof}
Being $\{x_1,x_2,\dotsc\}$ a closed set, $\tilde x_1$ and $\tilde x_2$ exist. 
If $\mu=x_j$ for some $j$ then $\tilde x_1=\tilde x_2$ and one can just apply the Jensen inequality. 
So, let $\tilde x_1$ and $\tilde x_2$ be distinct and for simplicity let them be equal to $x_1$ and $x_2$ with $x_1<x_2$. 

To begin with, we consider the particular case of a random variable $X$ with finite support $\{x_1,\dotsc,x_m\}$, which can be dealt by induction. To this aim, let us start with $m=3$, so that the random variable $X$ has support $\{x_1,x_2,x_3\}$. 
In this case, $x_1< \mu < x_2$, being $x_1$ and $x_2$ 
closest to $\mu$ than $x_3$. Hence, either $x_3>x_2>x_1$ or $x_3<x_1<x_2$. 
In the first case $x_2$ is a linear combination of $x_1$ and $x_3$, in the second case $x_1$ is a linear combination of $x_3$ 
and $x_2$. Let us consider the first case, being the second one just about swapping $x_1$ and $x_2$, 
and take $s= (x_3-x_2)/(x_3-x_1)$, so that $x_2=sx_1+(1-s)x_3$. Note that $s<1$ and therefore 
using Jensen inequality one can write:
\[\begin{split}
t_1\phi(x_1)&+t_2\phi(x_2)+t_3\phi(x_3)\\
&=\left[t_1 - \frac{s}{1-s}t_3\right]\phi(x_1)+t_2\phi(x_2)+\frac{t_3}{1-s} \left[s \phi(x_1)+(1-s)\phi(x_3)\right]\\
&\geq \left[t_1-\frac{s}{1-s}t_3\right]\phi(x_1)+\left[t_2+\frac{1}{1-s}t_3\right]\phi(x_2),
\end{split}\] 
where 
\[
	\left[t_1-\frac{s}{1-s}t_3\right]x_1+\left[t_2+\frac{1}{1-s}t_3\right]x_2=t_1x_1+t_2x_2+t_3x_3=\mu,
\] 
and therefore $t_1-st_3/(1-s)=t$ and $t_2+t_3/(1-s)=1-t$. 
This proves that the statement of this lemma is true if $X$ is concentrated on three points. 
In order to complete the proof by induction, 
fix $m\geq 3$ and assume that it is true for any random variable with finite support $\{x_1,\dotsc,x_m\}$. 
Hence, for any random variable $X$ with support $\{x_1,\dotsc,x_{m+1}\}$, one has that 
\[
	\mean[\phi(X)\mid X\in \{x_{m-1},x_m,x_{m+1}\}] \geq u\phi(x_{i_1})+(1-u)\phi(x_{i_2}),
\]
for some $u$ in the unit interval, and some pair $(i_1,i_2)$ of distinct elements of $\{m-1,m,m+1\}$ 
such that
\[
	ux_{i_1}+(1-u)x_{i_2}=\mean[X\mid X\in \{x_{m-1},x_m,x_{m+1}\}],
\]
since we have proved the statement is true for random variables concentrated on three points such as one 
whose distribution is the conditional distribution of $X$ given that $X$ is either $x_{m-1},x_m$ or $x_{m+1}$. 
Without loss of generality, assume that $i_1=m-1$ and $i_2=m$. Therefore, letting $Y$ be a random variable taking values $\{x_1,\dotsc,x_m\}$ with probabilities 
\[
	\{\prob(X=x_1), \dotsc, \prob(X=x_{m-2}),u\prob(X\in\{x_{m-1},x_m,x_{m+1}\}),(1-u)\prob(X\in\{x_{m-1},x_m,x_{m+1}\})\},
\]
we have
\begin{equation}\label{eq: pr_lemma1}
	\begin{split}
	\mean[\phi(X)] &= \mean\Big[\phi(X)\ind_{\{X\in \{x_{m-1},x_m,x_{m+1}\}\}}\Big] +\sum_{j= 1}^{m-1} \phi(x_j)\prob(X=x_j)\\
	&\geq \prob(X\in\{x_{m-1},x_m,x_{m+1}\}) \big[u\phi(x_{m-1})+(1-u)\phi(x_{m})\big] + \sum_{j= 1}^{m-1} \phi(x_j)\prob(X=x_j)\\ 
	&=\mean[\phi(Y)]. 
	\end{split}
\end{equation}
Being $Y$ concentrated on $\{x_1,\dotsc,x_m\}$ with the same mean of $X$, one can use now the induction hypothesis 
to obtain that $\mean[\phi(Y)]$ is greater than or equal to $\mean[\phi(X)]\geq t\phi(x_1)+(1-t)\phi(x_2)$, and by 
\eqref{eq: pr_lemma1} the same is true for $\mean[\phi(X)]$. 

At this stage we have proved the statement for all random variables with finite support. We shall now consider the infinite case. Let $x_{(1)},x_{(2)},\dotsc$ be the decreasing rearrangement of $x_1,x_2,\dotsc$ and let 
\[
	Y_m=X\ind_{\{X< x_{(m)}\}}+\mean[X\mid X\geq x_{(m)}]\ind_{\{X\geq x_{(m)}\}},
\] 
where $m\geq 2$ is large enough so that $x_1$ and $x_2$ are both smaller than $x_{(m)}$. At this stage,  
Jensen inequality can be applied obtaining
\begin{align*}
\mean[\phi(X)] &=\sum_{j= 1}^{m-1} \prob(X=x_{(j)})\phi(x_{(j)}) +\prob(X\geq x_{(m)}) \sum_{j= m}^\infty \phi(x_{(j)})\frac{\prob(X=x_{(j)})}{\prob(X\geq x_{(m)})}\\
&\geq \sum_{j= 1}^{m-1} \prob(X=x_{(j)})\phi(x_{(j)}) + 
\prob(X\geq x_{(m)})\phi\left( \displaystyle{\sum_{j= m}^\infty} x_{(j)} \frac{\prob(X=x_{(j)})}{\prob(X\geq x_{(m)})} \right)\\
&=\mean[\phi(Y_m)]. 
\end{align*}
Clearly, $Y_m$ is a random variable  
valued into $\{x_{(1)},\dotsc,x_{(m-1)},\mean[X\mid X\geq x_{(m)}]\}$, where $\mean[X\mid X\geq x_{(m)}]$
is not closer to the mean than $x_1$ or $x_2$ since it is smaller than both of them. 
Moreover, the mean of $Y_m$ is the same of $X$. 
Therefore, having already proved that the statement holds for random variables with finite support,
we can conclude that $\mean[\phi(Y_m)]\geq t\phi(x_1)+(1-t)\phi(x_2)$ as desired. 

\end{proof}

We are now ready to derive the minimum normalized variance for the \textsc{SDM} model. Consider the function $\phi:\BR\to\BR$ of the form
\[ 
	\phi_\theta(u)=\frac{2u(1-u)^3}{(\theta+3-2u)(\theta+2-u)}+u^2, 
\]
and let $\tau_K$ denote $\mean[S_\textsc{SDM}\mid K]$. Being $\mu=\mean[S_\textsc{SDM}]=\mean[\tau_K]$, Equation~\eqref{eq: var_a_fisso} and Equation~\eqref{eq: mean} imply that 
\[
	\var(S_\textsc{SDM})=\mean[\phi_\theta(\tau_K)+2 \mu \tau_K+\mu^2].
\]
Hence, it suffices to show that the function $u\to \phi_\theta(u)+2\mu u+\mu^2$ is convex to apply Lemma \ref{l: convex}. Specifically, we need to show the convexity of $\phi_\theta$, which in turn is given if the convexity of the function 
\[
	\psi_\theta(u)=\frac{\phi_\theta(1-(\theta+1)u)}{\theta+1}+\frac{2u}{1+\theta}
\]
is shown. The function $\psi_\theta$ is given by:
\[     \psi_\theta(u)= 2\frac{[1-(\theta+1)u]u^3}{(1+2u)(1+u)}+\frac{[1-(\theta+1)u]^2}{\theta+1}+\frac{2u}{1+\theta}
           =\psi_{(1)}(u)+(\theta+1)\psi_{(2)}(u)+\psi_{(3)}(u),            \]
where 
\[
	\psi_{(1)}(u)=\frac{2}{(1+2u)(1+u)}, \qquad \psi_{(2)}(u)=u^2-\frac{2u^4}{(1+2u)(1+u)}, \qquad \psi_{(3)}(u)=\frac{2u}{1+\theta},
\]
are all convex functions on the unit interval (convexity of $\psi_{(2)}$ can be tested numerically), and therefore $\psi_\theta$ is also a convex function. 

Being $\var(S_\textsc{SDM})$ a convex function of $\tau_K$, which is a discrete random variable, by Lemma \ref{l: convex}, 
the support of the prior distribution of $K$ that  minimizes $\var(S_1)$ is made of the two integers on which the function 
$k\to \tau_K$ is closest to $\mu$, i.e.\ $\{\bar k, \bar k +1\}$, with 
\[
	\bar k=\bar k(\theta,\mu)=\left\lfloor \frac{(\theta+\mu)}{\theta(1-\mu)} \right\rfloor.
\] 
In order to ensure that $\mean[\tau_K]=\mu$, one has that 
\[ 
	\pi(\{\bar k +1\})=1-\pi(\{\bar k\})=p(\theta,\mu),
\] 
where 
\[
	p(\theta,\mu)=\left[1-(1-\mu)\frac{(\theta\bar k(\theta,\mu)+1)}{(\theta+1)}\right]\left(\frac{1}{\theta}+1+\bar k(\theta,\mu)\right).
\]  
Therefore, if we set  
\[
	\zeta(k,\theta)=2 \frac{\theta(\theta+1)(k-1)}{(k\theta+3)(k\theta+2)(k\theta+1)^2},
\] 
so that $\var(S_\textsc{SDM}\mid K)=\zeta(K,\theta)$, then $\var(S_\textsc{SDM})=V_\textsc{SDM}(\theta,\mu)$, where:
\begin{align*}
	V_\textsc{SDM}(\theta,\mu)&=\zeta\big(\bar k(\theta,\mu),a\big)+p(\theta,\mu)\left[\zeta\big(1+\bar k(\theta,\mu),\theta\big)-\zeta\big(\bar k(\theta,u),\theta\big)\right]\\
	&+\frac{(\theta+1)^2\theta^2}{\left[\big(1+\theta\bar k(\theta,\mu))(\theta\bar k(\theta,\mu)+\theta+1\big)\right]^2}p(\theta,\mu)\big(1-p(\theta,\mu)\big).
\end{align*}
In order to restrict the numerical search of the minimum in a bounded interval, 
the minimum normalized variance for model has been obtained as
\[
	\min_{x\in[0,1]} V_\textsc{SDM}(-\log(x),\mu).
\]
Such minimum has been computed as $\mu$ varies over the grid of 100 equally spaced points in the unit interval, 
and then plotted in Figure 2 of the manuscript. 

\end{itemize}
%
%

\section{Prior distribution for the species richness}
\label{s: priors_fin}

In Section 6 of the manuscript we consider two prior specifications for the species' richness, for the \textsc{SDM} and \textsc{DDM} models respectively. Here we report summaries of those priors, along with other possible specifications for the \textsc{DDM} richness distribution. 

\subsection{Richness prior for the \textsc{SDM} models}\label{sec:SDMprior}

In Section 6 of the paper, we consider the following prior distribution for the species richness with \textsc{SDM} models
\begin{equation}\label{eq: prior_mod1}
\pi(k)=\frac{\Ga(2-\psi_1)}{\Ga(\psi_2)\Beta(2-\psi_1-\psi_2,\psi_1)} \frac{\Ga(k+\psi_2-1)\Ga(k+\psi_1 -1)}{k!(k-1)!},\quad k\geq 1,
\end{equation}
where $\psi_1,\psi_2>0$ 
and $\psi_1+\psi_2<2$ and $\Beta$ denotes the beta function, i.e. $\Beta(x,y)=\Ga(x)\Ga(y)/\Ga(x+y)$, 
for every $x,y>0$. Some simplifications on the richness distribution summaries occur in the case of the conditional distribution of $(P_1,\dotsc,P_{K})$ given $K$  being uniform on the $K$--dimensional simplex, i.e., a symmetric Dirichlet distribution with concentration parameter one. 
%
Specifically, the prior mean equals
\begin{equation}\label{eq: mean_mod1}
 \mean[S_\textsc{SDM}]=1-2\mean\left[\frac{1}{1+K}\right].
\end{equation}
We now prove that the prior mean with the prior distribution given in Equation~\eqref{eq: prior_mod1} results in
\begin{equation}\label{eq: mean_gnedin}
 \mean[S_\textsc{SDM}] = \frac{\psi_1\psi_2}{(2-\psi_1)(2-\psi_2)}.
\end{equation}
We first recall that the probability mass function of a beta negative binomial with parameters $\al, \be, r>0$ is given by
\begin{equation*}
 q(k)=\frac{\Ga(r+k-1)\Beta(\al+r,\be+k-1)}{\Ga(r)\Beta(\al,r)(k-1)!}, \qquad k = 1, 2, \dots
\end{equation*}
Since $\tsum_{k\geq 1}q(k)=1$, we then have
\begin{equation}\label{eq: seriesbnb}
 \sum_{k\geq 1} \frac{\Ga(r+k-1)\Ga(\be+k-1)}{(k-1)!\Ga(\al+\be+r+k-1)}
 =\frac{\Ga(r)\Ga(\be)\Ga(\al)}{\Ga(\al+r)\Ga(\al+\be)}
\end{equation}
for every $\al,\be, r>0$. Hence, we obtain that
\[
	\mean\left[\frac{1}{1+K}\right] = \frac{\Ga(2-\psi_1)}{\Ga(\psi_2)\Beta(2-\psi_1-\psi_2,\psi_1)} \sum_{k\geq 1}\frac{\Ga(k+\psi_2-1)\Ga(k+\psi_1 -1)}{(k+1)!(k-1)!},
\]
where the series in the right--hand side of the previous equation is equal to the expression in Equation~\eqref{eq: seriesbnb} 
with $\al=3-\psi_1-r$. Therefore, 
\[
 \mean\left[\frac{1}{1+K}\right]=\frac{2-\psi_1-r}{(2-\psi_1)(2-r)},
\]
which yields to Equation~\eqref{eq: mean_gnedin} by Equation~\eqref{eq: mean_mod1}.

In order to compute the prior variance of the Gini--Simpson index, under the previous distributional assumption, note that being 
the conditional distribution of $(P_1,\dotsc,P_{K})$ given $K$ symmetric Dirichlet with parameter one, 
\[ 
	\mean\left[\sum_{j=1}^K P_j^2 \Big\vert K\right]^2= \frac{8}{(K+3)(K+2)(K+1)}+\frac{4}{(K+2)(K+1)}, 
\] 
which, by Equation~\eqref{eq: prior_mod1}, yields
\[ 
	\mean\left[\sum_{j=1}^K P_j^2\right]^2= \frac{8(4-\psi_1-\psi_2)(3-\psi_1-\psi_2)(2-\psi_1-\psi_2)}{(4-\psi_1)(3-\psi_1)(2-\psi_1)(4-\psi_2)(3-\psi_2)(2-\psi_2)} +\frac{4(3-\psi_1-\psi_2)(2-\psi_1-\psi_2)}{(3-\psi_1)(2-\psi_1)(3-\psi_2)(2-\psi_2)}.
\]
Thanks to Equation~\eqref{eq: mean_gnedin}, the prior variance of the index is given by
\[
\begin{split}
	\var(S_\textsc{SDM})&=\var\left(\sum_{j=1}^K P_j^2\right)= \mean\left[\sum_{j=1}^K P_j^2 \right]^2 - \big(1-\mean[S_1]\big)^2\\
		&=2\big(1-\mean[S_\textsc{SDM}]\big)\frac{3-\psi_1-\psi_2}{(3-\psi_1)(3-\psi_2)}  \left[ 2\frac{4-\psi_1-\psi_2}{(4-\psi_1)(4-\psi_2)} +1\right] - \big(1-\mean[S_\textsc{SDM}]\big)^2.
\end{split}
\]
Therefore, the normalized variance of the index equals
\begin{equation*}
	\rvar(S_\textsc{SDM})=2\frac{3-\psi_1-\psi_2}{(3-\psi_1)(3-\psi_2)\mean[S_\textsc{SDM}]} \left[ 2\frac{4-\psi_1-\psi_2}{(4-\psi_1)(4-\psi_2)} +1\right]-\frac{1}{\mean[S_\textsc{SDM}]}+1.
\end{equation*}
Denoting by $\mu$ the prior mean $\mean[S_\textsc{SDM}]$, Equation~\eqref{eq: mean_gnedin} yields  
\begin{equation}\label{eq: r_gnedin} 
	\psi_2=\frac{2\mu(2-\psi_1)}{\be+\mu(2-\psi_1)}. 
\end{equation}
Hence, the normalized prior variance of the index becomes:
\begin{equation*}
\rvar(S_\textsc{SDM})=2\frac{\psi_1 + (2-\psi_1)[\psi_1+\mu(1-\psi_1)]}{(3-\psi_1)[3\psi_1+\mu(2-\psi_1)]\mu} 
\left\{ \frac{\psi_1(4-\psi_1)+\mu(2-\psi_1)^2}{(4-\psi_1)[2\psi_1+\mu(2-\psi_1)]} +1\right\}-\frac{1}{\mu}+1.
\end{equation*}
The variance can be assessed as large as desired, as for any fixed prior mean $\mu$, it converges to one as $\psi_1$ converges to zero. Indeed, 
the normalized prior variance approaches one as $\psi_2$ goes to $2$ and $\psi_1$ go to zero. In fact, in such limit case, the prior of $K$ approaches a distribution that is concentrated on one and infinity, being $K-1$ distributed as $\textsc{NegBinom}(2-\psi_2-\psi_1, \psi_1, \psi_2)$, that is a mixture of negative binomials where the mixing distribution over the success rate is a $\textsc{Beta}(2-\psi_2-\psi_1, \psi_1)$.  
This happens because as $2-\psi_2-\psi_1$ and $\psi_1$ go to zero, a $\textsc{Beta}(2-\psi_2-\psi_1, \psi_1)$ distribution approaches a distribution that is concentrated on $\{0,1\}$. Moreover, in such limit case the prior mean can be freely 
assessed in the unit interval, since if $\psi_2=2-\psi_1/\mu$ with $0<\mu\leq 1$, and $\psi_1$ converges to zero, 
then the prior mean \eqref{eq: mean_gnedin} converges to $\mu$, whereas the normalized prior variance converges to one. 


\subsection{Richness prior for the \textsc{DDM} model}
As done in the manuscript, let
 \begin{equation}\label{eq: beta geometric}
\pi(k)=\frac{(\ga_1+2)(\ga_1+1)}{[\ga(1+\ga_2)+1](1+\ga_2)}
\frac{(\ga_1\ga_2/2)_{k-1}k^2}{(2+\ga_1(1+\ga_2))_{k}}, \qquad \ga_1,\ga_2>0, 
\end{equation}
where $(a)_{b}$ denotes the rising factorial, i.e. $(a)_{b}=\Ga(a+b)/\Ga(a)=a(a+1)\dots (a+b-1)$,  for every integer $b\geq 1$ and $(a)_{0}=1$, for every $a>0$. The following identities hold true, whose proofs are given in Section \ref{s: proofs_formulas}.
\begin{align} 
\nonumber
\mean\left[\frac{1}{K}\right]&=\frac{1}{1+\eta}, \qquad \mean\left[\frac{1}{K^2}\right]=\frac{2(1+\ga_1)}{(1+\ga_2)[\ga_1(2+\ga_2)+2]},\\
\label{eq: mean_nice_model}
\mean[S_\textsc{DDM}] &= \frac{\atwo}{1+\atwo}\frac{\ga_2}{1+\ga_2},\\
\begin{split}
\label{eq: var_nice_model}
\var(S_\textsc{DDM})&=
\smallfrac{2\atwo}{(\atwo+3)(\atwo+2)(\atwo+1)^2}
\smallfrac{\ga_2}{1+\ga_2}
\left[1+\smallfrac{\atwo\ga_1}{\ga_1(2+\ga_2)+2} \right]+\left(\smallfrac{\atwo}{1+\atwo}\right)^2
\smallfrac{(\ga_1+2)\ga_2}{(1+\ga_2)^2[2+(2+\ga_2)\ga_1]}.\end{split}
\end{align}
Both parameter $\ga_1$ and $\ga_2$ affect the heaviness of the tail. Using Stirling formula, i.e., $\Ga(x)\sim  \sqrt{2\pi} x^{x-1/2}e^{-x}$ as $x\to \infty$, 
one obtains that $\Ga(x)/\Ga(x+y)\sim x^{-y}$ as $x\to \infty$ for every fixed $y$, 
and therefore Equation~\eqref{eq: beta geometric} yields
\[ 
\pi(k) \sim 
c_{\ga_1,\ga_2}
k^{-\ga_1(1+\ga_2/2)-1}
\]
as $k$ diverges to infinity, where $c_{\ga_1,\ga_2}$ is a suitable constant. 

It is not difficult to derive derive the normalized variance of the index from \eqref{eq: var_nice_model}, which equals
\begin{equation*}
	\rvar(S_\textsc{DDM})= \frac{2(1+\ga_2)}{(\atwo+3)(\atwo+2)(\atwo+\ga_2+1)}\left[1+\frac{\atwo\ga_1}{\ga_1(2+\ga_2)+2}\right] +\frac{\atwo(\ga_1+2)}{[2+(2+\ga_2)\ga_1](\atwo+\ga_2+1)}.
\end{equation*}
Being
\begin{equation}\label{eq: relvar_bnb1}
 \rvar(S_\textsc{DDM})=1-\left(1-\frac{2}{(\atwo+2)(\atwo+3)}\right)
 \frac{1+\ga_2}{\atwo+\ga_2+1}
 \left\{ 1+\frac{\atwo}{2+\ga_2+2/\ga_1} \right\},
\end{equation}
the normalized variance is a decreasing function of $\ga_1$. The limit case of $\ga_1$ being equal to zero leads to the equality of $\mean[1/K]$ and 
$\mean[1/K^2]$, i.e. $\pi(k)$ being concentrated on one and infinity.
In such limit case, the normalized variance is equal to
\[
\rvar(S_\textsc{DDM}\mid\gamma_1 = 0)= 1-\left(1-\frac{2}{(\atwo+2)(\atwo+3)}\right)
 \frac{1+\ga_2}{\atwo+\ga_2+1},
\]
which converges to one as $\atwo$ diverges to infinity. On the opposite side, in the limit case of $\ga_1=\infty$ the normalized variance reaches its minimum, that is 
\begin{equation}\label{eq: rvar_etainf}
	\rvar(S_\textsc{DDM} \mid \gamma_1 = \infty)= 1-\left(1-\frac{2}{(\atwo+2)(\atwo+3)}\right)\frac{1+\ga_2}{2+\ga_2}\left(1+\frac{1}{\atwo+\ga_2+1}\right), 
\end{equation}
which is bigger than or equal to the normalized prior variance with degenerate $K$, with the equality in place only in the case of $\atwo=0$ or $\gamma_2=\infty$. If both $\atwo$ and $\gamma_1$ are infinity, then 
\[
	\rvar(S_\textsc{DDM}\mid \gamma_1 = \infty, \theta = \infty)=\frac{1}{2+\gamma_2}=1-\frac{1}{2-\mean[S_2]}
\] 
being $\gamma_2=\mean[S_2]/(1-\mean[S_2])$ by \eqref{eq: mean_nice_model}. 
Therefore, if $\atwo=\ga=\infty$ and $\mean(S_2)>1/2$, then $\rvar(S_2)<1/3$, i.e.\ the normalized prior variance of the index is smaller than the variance of the uniform distribution. Moreover, if the prior mean of the index is fixed to a value $\mu$, then the minimum of the relative variance as $\atwo$ varies is smaller than the relative variance obtained with the Dirichlet process, i.e. 
\[
	\rvar(S_\textsc{DP}) = \frac{2}{(\atwo+2)(\atwo+3)}=\frac{2(1-\mu)^2}{(2-\mu)(3-2\mu)},
\] 
using \eqref{eq: condmean} and \eqref{eq: condnvar} with $K$ degenerate and infinite. Indeed, if the prior mean is fixed to be $\mu$, with $\atwo=\mu/(1-\mu)$, and we are having infinite $\gamma_2$, the relative variance \eqref{eq: rvar_etainf} is equal to 
\[
	\rvar(S_\textsc{DDM}\mid \gamma_2 = \infty) = \frac{2(1-\mu)^2}{(2-\mu)(3-2\mu)}
\]
which corresponds to the Dirichlet process case.  It is easy to verify empirically that the minimum of the relative variance is not just smaller but strictly smaller than such value. To this aim one can look at Figure~\ref{fig: rvar_ex_betag}, where the relative variance is plotted as a function of $\atwo$ where $\atwo$ ranges in the interval $[\mu/(1-\mu),\infty)$.  Indeed, consider that by \eqref{eq: mean_nice_model}, $\gamma_2=\mu(1+\atwo)/(\atwo-\atwo\mu-\mu)$ and being $\gamma_2$ positive, it must be $\atwo\geq \mu/(1-\mu)$.  
\begin{figure}[ht]\begin{center}
\includegraphics[width=.7\textwidth]{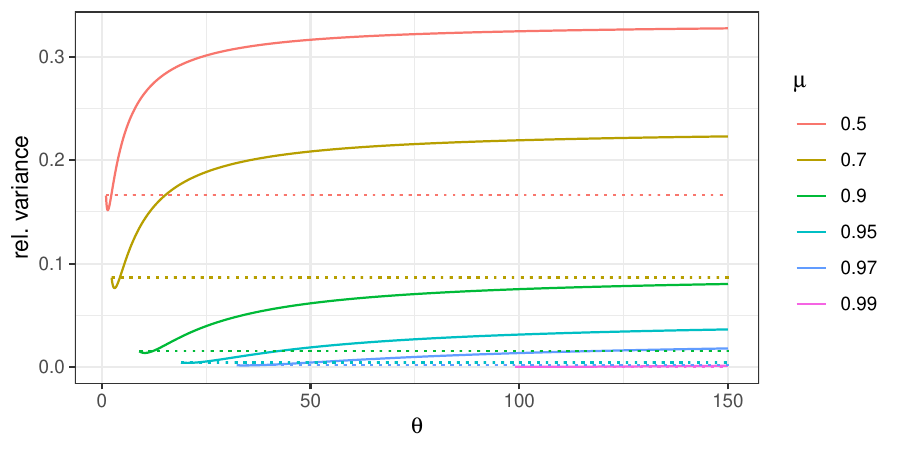}
\end{center}\caption{Normalized variance with the \textsc{DDM} model and the prior \eqref{eq: beta geometric}. Filled lines correspond to the model specification with infinite $\ga_1$ and fixed prior mean. Dotted lines correspond to the case with infinite $\gamma_2$, which equals the Dirichlet process case. }
\label{fig: rvar_ex_betag}\end{figure}
We finally remark that the normalized variance is  decreasing function of $\gamma_2$, whereas the mean is an increasing function of it. 
Indeed, \eqref{eq: relvar_bnb1} 
is equal to
\begin{equation*}
\begin{split}
\rvar(S_\textsc{DDM})&=1-\left(1-\frac{2}{(\atwo+2)(\atwo+3)}\right)
 \left(1-\frac{\atwo}{1+\gamma_2+\atwo}\right)
 \left( 1+\frac{\atwo}{2+\gamma_2+2/\gamma_1} \right)\\
 &=1-\left(1-\frac{2}{(\atwo+2)(\atwo+3)}\right)
 \left(1-\frac{2\atwo(1+1/\gamma_1)}{(2+\gamma_2+2/\gamma_1)(1+\gamma_2+\atwo)} \right),
\end{split}
\end{equation*}
which is clearly a decreasing function of $\gamma_2$, for any fixed $\gamma_1$. 
Instead, the mean is increasing as function of $\atwo$. The normalized variance converges to $1/3$ if $\atwo$ goes to zero and converges 
to $(2+\gamma_1)/\big[2+\gamma_1(2+\gamma_2)\big]$ if $\atwo$ diverges to infinity.




\subsection{Some other priors based on the \textsc{DDM} model}\label{sec:example priors}

In the following examples, we consider priors of $K$ for which there is a closed expression for the prior mean and variance 
of the index. The proofs of the formulas given in these examples are provided in Section \ref{s: proofs_formulas}. 
\begin{itemize}
\item[-] \textbf{Example 1} - Let
\[
\pi(k)=\smallfrac{e^{-\la}}{\la^2+3\la+1}\smallfrac{\la^{k-1}k^2}{(k-1)!},\qquad \la>0.
\]
Such a prior of $K$ is unimodal. Specifically, tthe mode is the integer part of the solution of the equation $(x-1)^3/x^2=\la$. Then, a priori, we have
\begin{align*}
\mean\left[\smallfrac{1}{K}\right]&=\smallfrac{1+\la}{\la^2+3\la+1}, \qquad \mean\left[\smallfrac{1}{K^2}\right]=\smallfrac{1}{\la^2+3\la+1},\qquad
\mean[K]=\smallfrac{\la^3+6\la^2+7\la+1}{\la^2+3\la+1},\\
\mean(S_\textsc{DDM})&=\smallfrac{\atwo}{1+\atwo}\left(1-\smallfrac{\la+1}{\la^2+3\la+1}\right),\\
\var(S_\textsc{DDM})&=\smallfrac{2\atwo}{(\atwo+3)(\atwo+2)(\atwo+1)^2}\left[1+\smallfrac{(\atwo-1)\la-1}{\la^2+3\la+1}\right]+
\left(\smallfrac{\atwo}{1+\atwo}\right)^2\smallfrac{\la}{(\la^2+3\la+1)^2}.
\end{align*}
\item[-]  \textbf{Example 2} - We now consider
\[
\pi(k)=\smallfrac{(1-p)^{r+1}}{1+p(r-1)}\binom{k+r-2}{k-1}p^{k-1}k, \qquad r>0,\, r\neq 1, \quad 0<p<1,
\]
Such a prior is unimodal and the mode is the integer part of 
\[
	1+\left[r+\sqrt{r^2+4(r-1)(1-p)/p}\right]\frac{p}{2(1-p)}.
\]
Hence, a priori for the Gini--Simpson index we have
\begin{align*}
\mean[S_\textsc{DDM}]&=\smallfrac{\atwo}{1+\atwo}\left[\smallfrac{pr}{1+p(r-1)}\right],\\
\var(S_\textsc{DDM})&=
\smallfrac{2a}{(\atwo+3)(\atwo+2)(\atwo+1)^2}\left[\atwo\left(1-\smallfrac{1-(1-p)^{r-1}}{p(r-1)\{1+p(r-1)\}}\right)
-(\atwo-1)\smallfrac{pr}{1+p(r-1)}\right]\\
&\phantom{XXX}+\left(\smallfrac{\atwo}{1+\atwo}\right)^2
\left[\frac{1-(1-p)^{r-1}}{p(r-1)\{1+p(r-1)\}}
-\smallfrac{(1-p)^2}{\{1+p(r-1)\}^2}
\right].
\end{align*}

\item[-]  \textbf{Example 3} - Lastly, we consider the following prior distribution
 
\begin{equation*}
\pi(k)=\smallfrac{p^3}{2-p}k^2(1-p)^{k-1}, \qquad 0<p<1. 
\end{equation*}
In this case, the distribution $\pi$ is unimodal and its mode is the integer part of $(1-\sqrt{1-p})^{-1}$, which is one 
if $p>3/4$ and tends to infinity as $p$ converges to zero. 

\begin{align*}
\mean\left[\frac{1}{K}\right]&=\smallfrac{p}{2-p}, \qquad \mean\left[\frac{1}{K^2}\right]=\smallfrac{p^2}{2-p},\\
\mean[S_\textsc{DDM}]&=\smallfrac{2\atwo}{\atwo+1}\left(1-\smallfrac{1}{2-p}\right),\\
\var(S_\textsc{DDM})&=\smallfrac{2\atwo}{(\atwo+3)(\atwo+2)(\atwo+1)^2}\left(1+(\atwo-1)\smallfrac{p}{2-p}-\atwo\smallfrac{p^2}{2-p}\right)+
\left(\smallfrac{\atwo}{1+\atwo}\right)^2\smallfrac{p^2(1-p)}{(2-p)^2}.
\end{align*}

\end{itemize}
One can assess prior mean and prior normalized variance 
with maximum freedom as long as the prior of $K$ can be concentrated only on one and infinity, for any value of $\mean[1/K]$, 
at least for a limit case. 
This is not a case with the distributions given in the previous examples, which in the limit case are concentrated either on one 
or infinite, but not on both. For this reason, one could consider as a prior of $K$ a mixture of two different 
distributions chosen among the previous examples. 
A good alternative is given by \eqref{eq: beta geometric}.

\subsection{Proofs of 
	\eqref{eq: mean_nice_model}--\eqref{eq: var_nice_model} 
and formulas of the examples in Section~\ref{sec:example priors}}\label{s: proofs_formulas}

In Section \ref{s: priors_fin} with \eqref{eq: beta geometric} and in Section~\ref{sec:example priors}, some priors are provided 
such that there is a closed expression for \eqref{eq: index_mean2} and \eqref{eq: var-mod2}. In general, it is convenient to consider a prior specification $\pi$ for the richness $K$ which admits the following representation
\[
	\pi(k)=\frac{k^jh(k)}{\mu_j},\qquad k=1,2,\dotsc,
\]  
where $j \in \{1, 2\}$, $h$ is a probability function on the positive integers and for every integer $l$, $\mu_l$ denotes the moment of $h$ of order $l$, 
i.e. $\mu_l=\tsum_{k=1}^\infty k^{l}h(k)$. Hence, if $j=1$
\[
	\mean\left[\frac{1}{K}\right]=\frac{1}{\mu_{1}}, \qquad \mean\left[\frac{1}{K^2}\right]=\frac{\mu_{-1}}{\mu_{1}},
\] 
while for $j=2$ we have 
\[
	\mean\left[\frac{1}{K}\right]=\frac{\mu_{1}}{\mu_{2}}, \qquad  \mean\left[\frac{1}{K^2}\right]=\frac{1}{\mu_{2}}.
\] 
Moreover, it is not difficult to compute the prior mean of $K$, which corresponds to $\mean[K]=\mu_2/\mu_1$ if $j=1$, and $\mean[K]=\mu_3/\mu_2$ if $j=1$. Clearly, we should consider probability functions $h$ for which the pair of moments $(\mu_{-1},\mu_1)$ or $(\mu_{1},\mu_2)$ not only are finite, but have also simple expressions. In the provided examples, we consider the probability function of the Poisson, negative binomial, geometric and beta geometric distribution. Each one of this probability function is shifted so that $h$ is supported only by the positive integers. In particular, for the prior given by \eqref{eq: beta geometric}, $j=2$ and $h$ is the $\textsc{Beta-Geometric}(\alpha, \beta)$ probability function with parameters $\al=\ga_1+2$ and $\be=\gamma_1\ga_2/2$.

\section{Computation of the posterior mean with the \textsc{SDM} model and
the prior given in Section~\ref{sec:SDMprior}}

With the \textsc{SDM} model, one can obtain the posterior mean mixturing \eqref{eq: posterior_mean2}, which takes form 
\begin{equation}\label{eq: postmeanM1a}
\mean\big[S_\textsc{SDM}\mid X_1, \dots, X_n\big] = n(n-1)\Big(\mu_{n}^{(0)}-\mu_{n}^{(1)}\Big){\hat S}_n + (n-1)(n+\theta+1)\mu_{n}^{(1)}-n(n+\theta)\mu_{n}^{(0)} + 1,
\end{equation}
where $\hat S_n$ is the frequentist unbiased estimator of the Gini--Simpson index and 
\[
	\mu_{n}^{(0)}= \sum_{k=k_n}^\infty \frac{\pi_n(k)}{\theta k+n},\qquad \mu_{n}^{(1)}= \sum_{k=k_n}^\infty \frac{\pi_n(k)}{\theta k+n+1},
\]  
with $\pi_n(k)$ denoting the posterior distribution of $K\mid X_1, \dots, X_n$, so that 
\[
	\mu_{n}^{(0)}-\mu_{n}^{(1)}= \sum_{k=k_n}^\infty \frac{\pi_n(k)}{(\theta k+n)(\theta k+n+1)}.
\] 
Hence, we can rewrite the posterior expectation of Equation~\eqref{eq: postmeanM1a} as
\begin{equation*}\label{eq: postmeanM1b}
	\begin{split}
	\mean\big[S_\textsc{SDM}\mid X_1, \dots, X_n\big] &= n(n-1)\Big(\mu_{n}^{(0)}-\mu_{n}^{(1)}\Big){\hat S}_n - (n-1)(n+\theta+1)\Big(\mu_{n}^{(0)}-\mu_{n}^{(1)}\Big)-(\theta+1)\mu_{n}^{(0)} + 1 \\
	&= 1-(n-1)\Big[n(1-{\hat S}_n)+\theta+1\Big]\Big(\mu_{n}^{(0)}-\mu_{n}^{(1)}\Big)-(\theta+1)\mu_{n}^{(0)},
	\end{split}
\end{equation*}
emphasizing that the Gini--Simpson estimate is higher when the posterior probability of the richness $K$ concentrates mass on larger values of its support, i.e., if $K$ is believed to be large. The expression of the Gini--Simpson index posterior mean can be further simplified, obtaining 
\begin{equation} \label{eq: postmean_gnedin}
	\mean\big[S_\textsc{SDM}\mid X_1, \dots, X_n\big] = q_n^{(1)}{\hat S}_n+\Big(1-q_n^{(1)}\Big)
\left\{1-\frac{\theta+1}{n}\left[\frac{n\mu_{n}^{(0)}+1}{1-n(n-1)\Big(\mu_{n}^{(0)}-\mu_{n}^{(1)}\Big)}-1\right]\right\},
\end{equation}
where $q_n^{(1)}=n(n-1)\Big(\mu_{n}^{(0)}-\mu_{n}^{(1)}\Big)$. By denoting
\[
	W^{(1)}_n=1-\frac{\theta+1}{n}\left\{\frac{n\mu_{n}^{(0)}+1}{1-n(n-1)(\mu_{n}^{(0)}-\mu_{n}^{(1)})}-1\right\},
\]
with some algebraic steps Equation~\eqref{eq: postmean_gnedin} yields the posterior expectation of the Gini--Simpson index with \textsc{SDM} prior reported in the paper. We highlight that the weight $q_n^{(1)}$ approaches its minimum of zero when $\theta$ goes to infinite or if $K$ is degenerate at infinite, whereas approaches its maximum that is $1-2/(n+1)$ if $\theta$ goes to zero. 

If in particular, we consider the the prior given by \eqref{eq: prior_mod1}, with $\theta = 1$ then the posterior mean is given by \eqref{eq: postmean_gnedin} with 
\begin{align}\label{eq: m0-1}
\mu_{n}^{(0)}-\mu_{n}^{(1)}&=
\smallfrac{(3+n-\psi_1-\psi_2-k_n)(2+n-\psi_1-\psi_2-k_n)}{(n+2-\psi_1)(n+1-\psi_1)(n+2-\psi_2)(n+1-\psi_2)},\\
\label{eq: m0}
\mu_{n}^{(0)}&=
\smallfrac{2+n-\psi_1-\psi_2-K_n}{(n+1-\psi_1)(n+1-\psi_2)},
\end{align}
where $k_n$ is the number of distinct species observed in $X_1, \dots, X_n$. Since in this case  $W^{(1)}_n$ has closed form expression, one can easily see that $W^{(1)}_n$ is an increasing function 
of $k_n$ and of $\mean[S_\textsc{SDM}]$. Indeed, by \eqref{eq: postmean_gnedin}, $W^{(1)}_n$ is a decreasing function of both 
$\mu_{n}^{(0)}-\mu_{n}^{(1)}$ and of $\mu_{n}^{(0)}$, and by \eqref{eq: m0-1} $\mu_{n}^{(0)}-\mu_{n}^{(1)}$ 
is a parabola as a function of $K_n$ with vertex equal to $5/4-r-\beta+n$, which is bigger than $k_n$, and therefore decreasing. 
Similarly, $\mu_{n}^{(0)}-\mu_{n}^{(1)}$ and $\mu_{n}^{(0)}$ are decreasing functions of $\psi_2$, which in turn is a increasing function of $\mean[S_1]$ by \eqref{eq: r_gnedin} if $\psi_1$ is fixed. 


\begin{proof}[Proof of \eqref{eq: m0-1} and \eqref{eq: m0}]
Recall that if $W$ is a random variable such that $(W-1)\sim\textsc{NegBinomial}(\al, \be, \psi_2)$ 
with probability mass function
\[
	\phi(w)=\smallfrac{\Ga(\al+r)\Ga(\al+\be)}{\Ga(\psi_2)\Ga(\al)\Ga(\beta)} 
	\smallfrac{\Ga(w+\psi_2-1)\Ga(w+\beta -1)}{(w-1)!\Ga(\al+\be+\psi_2+w-1)},\quad w\geq 1.
\]
We are considering the \textsc{SDM} model with $\theta=1$ and the prior of $K$ given by \eqref{eq: prior_mod1}, i.e., the prior that lets $(K-1)\sim\textsc{NegBinomial}(2-\psi_1-\psi_2,\psi_1, \psi_2)$, 
where $\psi_1, \psi_2 > 0$ and $\psi_1+\psi_2<2$. 
By \eqref{eq: postK}, a posteriori we have $(K-k_n+1) \sim \textsc{NegBinomial}(2+n-\psi_1-\psi_2-k_n, \psi_1+k_n-1, \psi_2+K_n-1)$. 
Therefore, 
\begin{equation*}\begin{split}
	\mu_{n}^{(0)}-\mu_{n}^{(1)}&=\sum_{k=k_n}^\infty \frac{\pi_n(k)}{(k+n)(k+n+1)}
	=\sum_{h=1}^\infty \frac{\pi_n(h+k_n-1)}{(h+k_n+n-1)(h+k_n+n)}\\
	&=\smallfrac{(3+n-\psi_1-\psi_2-k_n)(2+n-\psi_1-\psi_2-k_n)}{(n+2-\psi_1)(n+1-\psi_1)(n+2-\psi_2)(n+1-\psi_2)},
\end{split}\end{equation*}
and
\begin{equation*}
\mu_{n}^{(0)}=\sum_{k=k_n}^\infty \frac{\pi_n(k)}{k+n}
=\sum_{h=1}^\infty \frac{\pi_n(h+k_n-1)}{h+k_n+n-1}
=\smallfrac{2+n-\psi_1-\psi_2-k_n}{(n+1-\psi_1)(n+1-\psi_2)}.
\end{equation*}
\end{proof}

\end{document}